\documentclass[aps,physrev,preprint, groupedaddress, showkeys]{revtex4-2}

\usepackage{amsmath}
\usepackage{amsfonts}
\usepackage{graphicx}
\usepackage{blkarray}
\usepackage{physics}
\usepackage{algorithm}
\usepackage{algpseudocode}
\usepackage{hyperref}
\usepackage{amssymb}
\usepackage{mathtools}
\usepackage{caption}
\usepackage{subcaption}
\usepackage{dsfont}
\usepackage{amsthm}

\newtheorem{Def}{Definition}
\newtheorem{problem}{Problem}
\newtheorem{corollary}{Corollary}
\newtheorem{lemma}{Lemma}
\newtheorem{conjecture}{Conjecture}
\newtheorem{remark}{Remark}

\newtheorem{theorem}{Theorem}

\begin{document}


\title{Optimizing Cost Hamiltonian Compilation for Max-Cut QAOA on Unweighted Graphs Using Global Controls and Qubit Bit Flips}


\author{Saber Dinpazhouh}
\email{sd118@rice.edu}
\affiliation{Computational and Applied Mathematics and Operations Research, Rice University, Houston, Texas, USA}

\author{Illya V. Hicks}
\email{ivhicks@rice.edu}
\affiliation{Computational and Applied Mathematics and Operations Research, Rice University, Houston, Texas, USA}


\date{\today}
\begin{abstract}
We study a cost Hamiltonian compilation problem for the quantum approximate optimization algorithm (QAOA) applied to the Max-Cut problem, focusing on trapped-ion quantum computers. Instead of standard compilation with CNOT and \(R_z\) gates, we employ global coupling operations and single-qubit bit flips. Prior work by Rajakumar et al.\ established that such a compilation is always possible.

Minimizing operational error requires short operation sequences. The problem reduces to a low-rank semi-discrete decomposition of the graph’s adjacency matrix, where the minimum achievable rank, the \textit{graph coupling number} \(gc(G)\), represents the number of global control layers. Rajakumar et al.\ introduced the \textit{Union of Stars} construction, proving \(gc(G) \leq 3n - 2\) for unweighted graphs with \(n\) vertices, and gave an \(O(m)\)-rank construction for weighted graphs.

We concentrate on unweighted graphs. We derive structural properties of the compilation problem and show the Union of Stars method is order-optimal by proving a lower bound of \(gc(G) \geq n - 1\) for a family of graphs. We also improve the general upper bound to \(2.5n + 2\). For particular graph families---cliques, perfect matchings, paths, and cycles---we provide sharper bounds. Further, we reveal a link between the problem and Hadamard matrix theory. Finally, we introduce a compact mixed-integer programming (MIP) formulation that outperforms the previously studied exponential-size MIP.

\end{abstract}


\maketitle

\section{Introduction \label{intro}}

Two primary limitations of current quantum computers are their limited number of qubits and the inherent noise in their operations. Although quantum hardware continues to improve in both qubit count and gate fidelity, the field remains far from realizing fully fault-tolerant quantum error correction~\cite{rajakumar2022generating}. John Preskill referred to this stage as the \emph{noisy intermediate-scale quantum} (NISQ) era, highlighting these limitations~\cite{Preskill2018quantumcomputingin}. 

Despite these constraints, several promising applications have been proposed for NISQ-era hardware, including quantum chemistry, the variational quantum eigensolver (VQE), and quantum optimization~\cite{Preskill2018quantumcomputingin}.

In the domain of combinatorial optimization, the quantum approximate optimization algorithm (QAOA) has emerged as a hybrid quantum-classical method capable of generating high-quality solutions to quadratic unconstrained binary optimization (QUBO) problems~\cite{farhi2014quantum}. Importantly, a broad class of combinatorial problems can be reformulated as QUBO instances~\cite{du2022solving,anthony2017quadratic,glover2022quantum}. These include classical problems such as Max-Cut, graph coloring, and set packing, as well as more general formulations like 0/1 integer programming, quadratic assignment, and constraint satisfaction problems~\cite{kochenberger2006unified}. As a result, QAOA is widely viewed as a general-purpose framework for tackling combinatorial optimization problems.

Among these, the Max-Cut problem has received particular attention as a benchmark for quantum optimization~\cite{moondra2024promise}. Given a weighted graph \( G = (V, E, z) \), Max-Cut seeks a partition of the vertex set \( V \) into two subsets \( S \) and \( V \setminus S \) that maximizes the total weight of edges crossing the partition. Max-Cut is one of Karp’s original 21 NP-hard problems~\cite{karp2009reducibility}, and is therefore believed to lack an efficient algorithm unless \( \text{P} = \text{NP} \).

To approximate Max-Cut using QAOA, each vertex is mapped to a qubit, and the computational basis states of the qubits encode the two vertex subsets~\cite{farhi2020quantum,bravyi2020obstacles}. The cut value is embedded in the energy of a cost Hamiltonian consisting of two-qubit terms for each edge \( e \in E \), where the optimal cut corresponds to the ground state of this Hamiltonian~\cite{rajakumar2022generating}.

Max-Cut QAOA was recently demonstrated on trapped-ion hardware using systems of tens of qubits~\cite{pagano2020quantum}. However, in these demonstrations, the problem graph matched the native hardware coupling graph~\cite{rajakumar2022generating}. In practice, this alignment is uncommon across most NISQ platforms, where arbitrary problem graphs must be compiled into the hardware's limited connectivity. This compilation step often results in significant circuit overhead, which in turn degrades the algorithm’s performance due to increased noise.

To implement Max-Cut QAOA for arbitrary graphs, the required cost Hamiltonian must be compiled into the available native gate set. The standard approach for this compilation involves using two controlled-NOT (CNOT) gates and one \( R_z \) rotation per edge of the input graph. Rajakumar et al.~\cite{rajakumar2022generating} proposed an alternative compilation method suitable for quantum spin systems, particularly trapped-ion platforms. Their method constructs arbitrary two-qubit interactions using a sequence of global entangling operations — namely, Mølmer–Sørensen (MS) gates — and individual qubit bit flips.

The goal of this alternative approach is to generate circuits with fewer layers, thereby reducing cumulative noise and improving QAOA performance. Figure~\ref{fig:standard-compile} and Figure~\ref{fig:new-compile} illustrate the standard and MS-based compilations, respectively, for a star graph with four vertices. It is worth noting that multiple MS-based compilation sequences can realize the same cost Hamiltonian. Among all such sequences, our focus is on minimizing the number of MS gates, which we treat as a proxy for circuit depth and noise accumulation.

While our focus is on trapped-ion systems, the compilation strategy is general and applies to other quantum architectures with long-range interactions, such as nuclear magnetic resonance (NMR)~\cite{vandersypen2001experimental}, Rydberg atom arrays~\cite{levine2018high,graham2019rydberg,henriet2020quantum}, and superconducting qubits coupled via a common bus resonator~\cite{xu2018emulating,onodera2020quantum}. Beyond improving Max-Cut QAOA implementations, such compilations may also facilitate more efficient quantum simulations~\cite{rajakumar2022generating}.

We investigate this compilation problem from both theoretical and computational perspectives. Focusing on unweighted graphs, we develop new bounds and algorithms that improve upon previous results.

The remainder of the paper is organized as follows. Section~\ref{sec:prelim} introduces preliminaries and establishes notation. Section~\ref{sec:problemdes} presents both the physical and mathematical formulations of the problem. The mathematical formulation is self-contained and may be read independently of the physical motivation.

Section ~\ref{sec:methods} contains our main results. Subsection~\ref{sec:meanrow} provides structural insights into the decomposition matrix \( \mathbf{P} \), while Subsection~\ref{sec:charac} characterizes graphs with small \( gc \) values. Subsection~\ref{sec:fundam} presents several foundational results.

In Subsection~\ref{sec:lower}, we introduce a novel lower-bound technique based on linear algebra, which is then applied in Subsection~\ref{sec:lower_specific} to specific families of graphs. Subsection~\ref{sec:upperspec} develops tight upper bounds for various graph classes. Notably, we show that for perfect matching graphs, the gap between our upper and lower bounds is exactly one. This result leads to an intriguing connection to the Hadamard matrix conjecture, discussed in Subsection~\ref{sec:hadamard}, where we show that resolving this gap is equivalent to resolving the conjecture.

In Subsection~\ref{sec:generupper}, we introduce a new compilation strategy—called the \emph{Union of Double-Stars}—that improves the general upper bound on \( gc(G) \) from \( 3n - 2 \) to \( 2.5n + 2 \).

Finally, in Subsection~\ref{sec:comp}, we propose a compact mixed-integer programming formulation (CMIPGC) and compare it numerically with the exponential-size MIP formulation of~\cite{rajakumar2022generating} in Section~\ref{sec:exp}. Our formulation uses theoretical upper bounds as warm starts and incorporates lower bound as a cutting plane, thereby tightly integrating theory and computation.

\begin{figure}
    \centering
    \includegraphics[width=\linewidth]{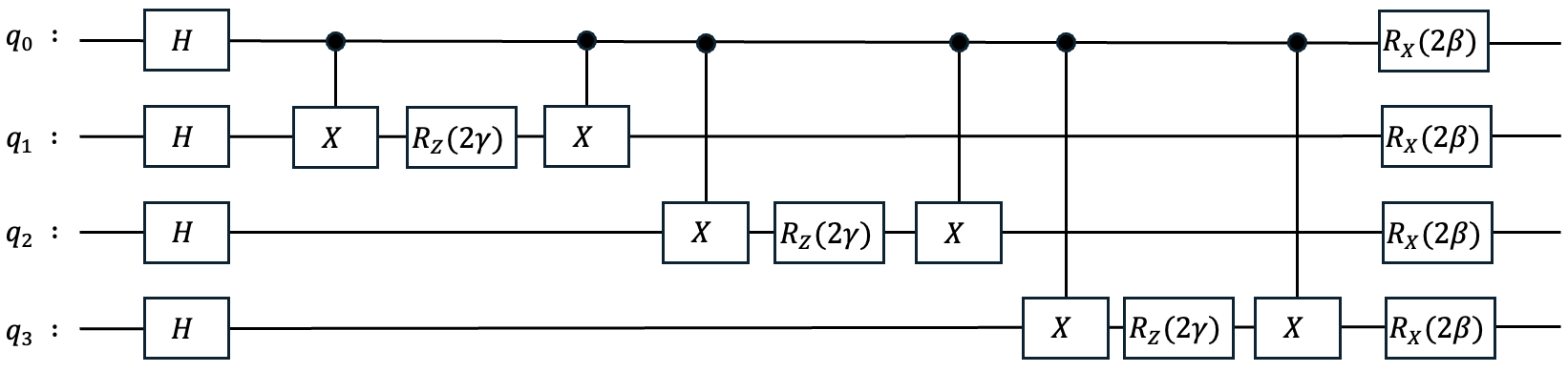}
    \caption{Standard QAOA Max-Cut compilation for a star graph on 4 vertices, where one central vertex is connected to the remaining three.}
    \label{fig:standard-compile}
\end{figure}

\begin{figure}
    \centering
    \includegraphics[width=\linewidth]{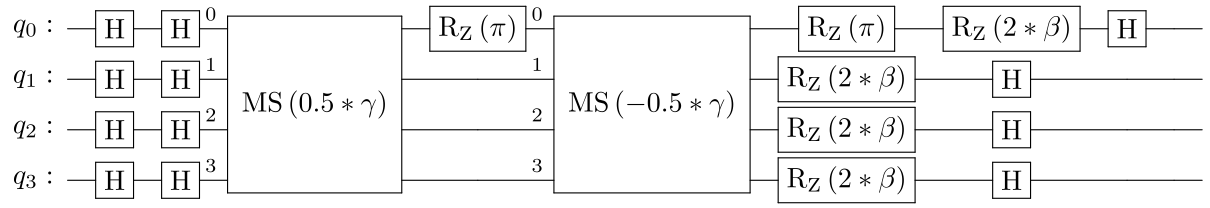}
    \caption{Example of an MS-based compilation for a star graph on 4 vertices, adapted from~\cite{rajakumar2022generating}.}
    \label{fig:new-compile}
\end{figure}

\section{Preliminaries \label{sec:prelim}}

We use \( \mathbb{N} \) to denote the set of natural numbers \( \{1, 2, 3, \dots\} \), and write \( [n] \) for the set \( \{1, 2, \dots, n\} \). Given a subset \( S \subseteq [n] \), we define the vector \( r_S \in \{-1, 1\}^n \) by setting
\[
(r_S)_i = \begin{cases}
-1 & \text{if } i \in S, \\
\phantom{-}1 & \text{otherwise}.
\end{cases}
\]
For instance, \( r_{\{1,2,3\}} = [-1, -1, -1, 1, 1]^\top \) for \( n = 5 \), and \( r_{\{1,4,6\}} = [-1, 1, 1, -1, 1, -1, 1]^\top \) for \( n = 7 \). We assume throughout that indexing begins at 1.

\medskip

Matrices are denoted with bold capital letters (e.g., \( \mathbf{A} \), \( \mathbf{P} \)). The all-ones vector is denoted \( \mathds{1} \), and the all-ones matrix by \( \mathbf{1} \). For a given matrix \( \mathbf{B} \), we write \( \mathbf{B}_i \) to refer to its \( i \)-th column and \( b_i \) for its \( i \)-th row.

\medskip

 Throughout this paper, we assume familiarity with basic notions from graph theory, including vertex degree, adjacency, complement graph, and complete graph. A graph \( G = (V, E) \) is called \emph{simple} if it contains no loops or multiple edges. For any subset \( S \subseteq V \), the \emph{induced subgraph} \( G[S] \) consists of the vertices in \( S \) and all edges in \( E \) with both endpoints in \( S \). A \emph{complete graph} (or \emph{clique}) on vertex set \( V \) is denoted \( K_V \), or simply \( K_{|V|} \). Given \( S \subseteq V \), we write \( \delta(S) \subseteq E \) for the set of \emph{cut edges}—those with exactly one endpoint in \( S \)—and \( \gamma(S) \subseteq E \) for the set of edges with both endpoints in \( S \).

A graph is \emph{bipartite} if its vertex set can be partitioned into disjoint subsets \( A \) and \( B \) such that every edge connects a vertex in \( A \) to one in \( B \). A \emph{biclique}, or \emph{complete bipartite graph}, is denoted \( K_{A,B} \) or \( K_{|A|, |B|} \). A \emph{star graph} is a biclique in which one part has size one, e.g., \( K_{1,4} \). A \emph{perfect matching} is a graph on an even number of vertices in which every vertex has degree one; we denote it by \( PM_q \) for a matching with \( q \) edges.

The \emph{path graph} \( P_n \) consists of vertices \( v_1, \dots, v_n \) with edges \( \{\{v_i, v_{i+1}\} : i \in [n-1]\} \). The \emph{cycle graph} \( C_n \) augments this path by including the edge \( \{v_n, v_1\} \). A disjoint union of cliques on vertex sets \( S_1, \dots, S_r \) is denoted \( K_{S_1} \cup \dots \cup K_{S_r} \), or equivalently \( K_{|S_1|} \cup \dots \cup K_{|S_r|} \).

\medskip
For standard graph theory terminology, we refer the reader to Bondy and Murty~\cite{bondy2008graph}. For linear algebraic notation and conventions, see Strang~\cite{strang2023linear}.

\section{Problem Description \label{sec:problemdes}}

\subsection{Physical Description \label{sec:physdes}}

Many quantum systems exhibit intrinsic (native) couplings between qubits. These spin–spin couplings are physically modeled by the zero-field Ising Hamiltonian:
\begin{equation}
    H_{\text{Ising}} = \sum_{i=1}^{n - 1} \sum_{j=i + 1}^{n} J_{i, j} \sigma_i^z \sigma_j^z,
    \label{eq:nativecoupling}
\end{equation}
where \( \sigma_j^z \) denotes the Pauli-Z operator acting on the \( j \)-th ion (and no other qubit). The term \( \sigma_i^z \sigma_j^z \) represents a ZZ coupling between qubits \( i \) and \( j \), and the coupling strengths \( J_{i, j} \in \mathbb{R} \) depend on the specific quantum hardware~\cite{rajakumar2022generating}.

\begin{Def}
A \textbf{coupling graph} \( G = (V, E, z) \), with \( |V| = n \), is a weighted graph abstracting a ZZ coupling operation on an \( n \)-qubit system, where each vertex \( v \in V \) represents a qubit, and the weight function \( z \) encodes the strength of ZZ couplings between qubits~\cite{rajakumar2022generating}.
\end{Def}

Under this abstraction, a general ZZ coupling operator
\[
C = \sum_{i=1}^{n - 1} \sum_{j=i + 1}^{n} a_{i,j} \sigma_i^z \sigma_j^z
\]
has a coupling graph whose adjacency matrix is symmetric, with entries \( a_{i,j} = a_{j,i} \in \mathbb{R} \). For example, the native coupling graph has an adjacency matrix \( \mathbf{J} \) as defined in Equation~\eqref{eq:nativecoupling}.

Rajakumar et al.~\cite{rajakumar2022generating} showed that an arbitrary ZZ coupling operator \( C \) can be implemented using a sequence of just two types of operations:
\begin{enumerate}
    \item \textbf{Global coupling} using the Mølmer–Sørensen interaction, which implements a ZZ coupling over all qubit pairs with a uniform strength \( W \in \mathbb{R} \), i.e.,
    \[
    MS = \sum_{i=1}^{n - 1} \sum_{j=i + 1}^{n} W \sigma_i^z \sigma_j^z.
    \]
    \item \textbf{Single-qubit bit flips}, implemented by applying the operator \( X_j = -i \sigma_j^x \), where \( \sigma_j^x \) is the Pauli-X operator on qubit \( j \).
\end{enumerate}

These operations interact through the identity:
\[
X_j^\dagger \sigma_i^z \sigma_j^z X_j = -\sigma_i^z \sigma_j^z.
\]
Thus, surrounding a global coupling operation with bit flips on qubit \( j \) effectively flips the sign of all ZZ terms involving qubit \( j \). By combining multiple such bit flips, one can control the sign pattern of couplings in the complete graph \( K_n \).

Let us denote by \( \mathbf{P} \in \{\pm 1\}^{k \times n} \) the matrix whose \( (p, j) \)-th entry is \(-1\) if a bit flip is applied to qubit \( j \) in the \( p \)-th global coupling operation, and \( +1 \) otherwise. Then, applying \( k \) global couplings, each with weight \( W_p \), results in the effective coupling operator:
\[
C = \sum_{p=1}^{k} \sum_{i=1}^{n - 1} \sum_{j=i + 1}^{n} W_p \mathbf{P}_{p,i} \mathbf{P}_{p,j} \sigma_i^z \sigma_j^z.
\]
Our goal is to synthesize a ZZ coupling operator whose coupling graph matches a given target graph with adjacency matrix \( \mathbf{A} \)~\cite{rajakumar2022generating}. This requires:
\[
\sum_{p=1}^{k} \sum_{i=1}^{n - 1} \sum_{j=i + 1}^{n} W_p \mathbf{P}_{p,i} \mathbf{P}_{p,j} \sigma_i^z \sigma_j^z = \sum_{i=1}^{n - 1} \sum_{j=i + 1}^{n} \mathbf{A}_{i,j} \sigma_i^z \sigma_j^z.
\]
Equating the coefficients of the ZZ terms on both sides yields \( \binom{n}{2} \) constraints:
\begin{equation}
    \mathbf{A}_{i,j} = \sum_{p = 1}^{k} W_p \mathbf{P}_{p,i} \mathbf{P}_{p,j}, \quad \forall\, 1 \le i < j \le n.
    \label{eq:probdesc}
\end{equation}

The overarching goal is to minimize the number of Mølmer–Sørensen gates, \( k \), among all circuits consisting solely of global \( ZZ \) couplings and local bit flips, such that the resulting operator implements the desired coupling graph under ideal (noiseless) conditions. A representative example of such a circuit is shown in Figure~\ref{fig:new-compile}.

\subsection{Mathematical Description \label{sec:mathdes}}

We now formalize the problem introduced in Equation~\eqref{eq:probdesc}. Let \( G = (V, E) \) be a simple, undirected, and unweighted graph with \( |V| = n \), and let \( \mathbf{A} \in \{0, 1\}^{n \times n} \) denote its adjacency matrix. Since \( G \) is undirected and loopless, \( \mathbf{A} \) is symmetric with zero diagonal entries.

We seek a matrix \( \mathbf{P} \in \{\pm 1\}^{k \times n} \) and a diagonal matrix \( \mathbf{W} \in \mathbb{R}^{k \times k} \) such that
\[
\mathbf{A} = \mathbf{P}^\top \mathbf{W} \mathbf{P} \odot \mathbf{J},
\]
where \( \odot \) denotes the Hadamard (element-wise) product and \( \mathbf{J} = \mathbf{1} - \mathbf{I} \) is the adjacency matrix of the complete graph \( K_n \).

The objective is to minimize \( k \), the number of Mølmer–Sørensen layers, or equivalently, the number of nonzero weights in \( \mathbf{W} \). We now formally state the Graph Coupling Problem.

\begin{problem}[Graph Coupling Problem]
\label{prob:gc}
\hfill\\
\textbf{Input:} A simple, undirected, unweighted graph \( G = (V, E) \) with adjacency matrix \( \mathbf{A} \in \{0, 1\}^{n \times n} \).

\[
\begin{aligned}
    \min_{\mathbf{P}, \mathbf{W}} \quad & \|\mathbf{W}\|_0 \\
    \text{subject to} \quad & \mathbf{A} = \mathbf{P}^\top \mathbf{W} \mathbf{P} \odot \mathbf{J}, \\
    & \mathbf{W} = \operatorname{diag}(w) \in \mathbb{R}^{k \times k}, \\
    & \mathbf{P} \in \{\pm 1\}^{k \times n}, \\
    & \mathbf{J} = \mathbf{1} - \mathbf{I}.
\end{aligned}
\]
\end{problem}

Rajakumar et al.~\cite{rajakumar2022generating} denoted the minimum value of this problem by \( gc(G) \), referring to it as the \emph{graph coupling number} of \( G \).

\section{Methods \label{sec:methods}}
In this section, we examine the theoretical foundations of the Graph Coupling Problem and subsequently introduce a compact mixed-integer programming formulation. Our discussion begins by considering two equivalent reformulations of the problem.

\subsection{Two Equivalent Reformulations of the Graph Coupling Problem \label{sec:meanrow}}

The Graph Coupling Problem was originally posed as a semi-discrete matrix decomposition. Here, we present two alternative but equivalent formulations that provide additional insight into the structure of the problem.

\medskip
\noindent
\textbf{Spin Biclique Decomposition.}  
The feasibility condition
\[
\mathbf{A} = \mathbf{P}^\top \mathbf{W} \mathbf{P} \odot \mathbf{J}
\]
can be expanded as
\begin{equation}
\mathbf{A} = \sum_{i=1}^{k} \mathbf{W}_{ii} \, p_i^\top p_i \odot \mathbf{J},
\label{eq:sb-decomp}
\end{equation}

where \( p_i \in \{\pm 1\}^n \) is the \( i \)-th row of \( \mathbf{P} \). We show that each matrix \( p_i^\top p_i \odot \mathbf{J} \) is the adjacency matrix of a special class of graphs, which we call \emph{spin bicliques}. Thus, the Graph Coupling Problem is equivalent to expressing the input graph as a weighted sum of a minimum number of spin bicliques.

\begin{Def}
Let \( G = (V, E, z) \) be a complete undirected graph with weight function \( z : E \to \{-1, 1\} \). We say \( G \) is a \textbf{spin biclique} if there exists a cut \( \delta(S) \), for some \( S \subseteq V \), such that
\[
z(e) = \begin{cases}
-1 & \text{if } e \in \delta(S), \\
\phantom{-}1 & \text{otherwise}.
\end{cases}
\]
We denote such a graph by \( SB_{|S|, |V \setminus S|} \), or by \( SB_{S, V \setminus S} \) when the bipartition is to be made explicit.
\end{Def}

We now formalize the connection between spin bicliques and vectors in \( \{\pm 1\}^n \).

\begin{lemma}
\label{thm:rowsb}
Let \( p \in \{\pm 1\}^n \). Then:
\begin{enumerate}
    \item \( p p^\top \odot \mathbf{J} \) is the adjacency matrix of a spin biclique.
    \item If \( \mathbf{A} \) is the adjacency matrix of a spin biclique, then there exist exactly two vectors for \( p \) satisfying \( \mathbf{A} = p p^\top \odot \mathbf{J} \), and they are negatives of each other.
\end{enumerate}
\end{lemma}

\begin{proof}
1. Let \( p = [\mathds{1}_q \;\; (-\mathds{1})_{n-q}]^\top \). Then
\[
p p^\top \odot \mathbf{J} =
\begin{bmatrix}
\mathds{1}_q \\
-\mathds{1}_{n-q}
\end{bmatrix}
\begin{bmatrix}
\mathds{1}_q^\top & -\mathds{1}_{n-q}^\top
\end{bmatrix}
\odot \mathbf{J} =
\begin{bmatrix}
\mathbf{1} & \mathbf{-1} \\
\mathbf{-1} & \mathbf{1}
\end{bmatrix}
\odot \mathbf{J},
\]
which is the adjacency matrix of \( SB_{q, n - q} \).

2. Conversely, let \( \mathbf{A} \) be the adjacency matrix of a spin biclique defined by a cut \( S \subseteq V \), and assume the vertices in \( S \) appear first. Then
\[
\mathbf{A} = 
\begin{bmatrix}
\mathbf{1} & \mathbf{-1} \\
\mathbf{-1} & \mathbf{1}
\end{bmatrix}
\odot \mathbf{J} = p p^\top \odot \mathbf{J},
\]
for \( \mathbf{p} = [\mathds{1}_{|S|} \;\; -\mathds{1}_{n - |S|}] \). The vector \( -p \) produces the same result, and no other sign pattern yields the same adjacency matrix.
\end{proof}

\begin{remark}
Lemma~\ref{thm:rowsb} provides a bijection (up to sign) between vectors in \( \{\pm 1\}^n \) and adjacency matrices of spin bicliques via the map \( p \mapsto p p^\top \odot \mathbf{J} \).
\end{remark}

This correspondence allows us to view the rows of \( \mathbf{P} \) as encoding spin bicliques. Consequently, the Graph Coupling Problem reduces to selecting a minimal-weighted subset of such bicliques that sum to the input graph.

\medskip
\noindent
\textbf{Column-Wise Dot Product View.}
An alternative reformulation arises by expressing the feasibility condition in terms of the columns of \( \mathbf{P} \). The expression
\[
\mathbf{A} = \mathbf{P}^\top \mathbf{W} \mathbf{P} \odot \mathbf{J}
\]
is equivalent to

\begin{equation}
\mathbf{A}_{i,j} = (\mathbf{P}_i \odot \mathbf{P}_j)^\top w, \quad \forall\, i < j,
\label{eq:column-wise-dot}
\end{equation}

where \( \mathbf{P}_i \) is the \( i \)-th column of \( \mathbf{P} \), and \( w \) is the vector of diagonal entries of \( \mathbf{W} \).

\begin{remark}
Let \( (\mathbf{P}, \mathbf{W}) \) be a feasible solution for a given graph \(G\). Suppose the disjoint sets of vertices \( A, B \subseteq V(G) \) have identical columns \( \mathbf{a}, \mathbf{b} \in \{\pm 1\}^k \), respectively, in \( \mathbf{P} \). Then every edge \( \{x, y\} \) with \( x \in A, y \in B \) has weight \( w^\top (a \odot b) \), and every edge within \( A \) has weight \( w^\top (a \odot a) = w^\top \mathds{1} = \operatorname{tr}(\mathbf{W}) \). This reveals that vertices sharing a column class behave uniformly with respect to all other classes—a structural property that we will exploit in our novel construction, Union of Double-Stars, in Section~\ref{sec:generupper}.
\label{cor:colclass}
\end{remark}

\subsection{Characterization of Graphs with Small \texorpdfstring{$gc$}{gc} \label{sec:charac}}

Next, we characterize all graphs with \( gc \) equal to 1, 2, and 3.

We begin with the case \( gc(G) = 1 \), showing that only complete or empty graphs satisfy this condition.

\begin{theorem}
Let \( G \) be a simple unweighted graph. Then
\[
gc(G) = 1 \iff \exists\, n \in \mathbb{N} \text{ such that } G = K_n \text{ or } G = \overline{K_n}.
\]
\label{thm:gc=1}
\end{theorem}

\begin{proof}
The condition \( gc(G) = 1 \) implies that \( G \) is a scalar multiple of a single spin biclique.

Since \( G \) is unweighted, it cannot contain edges with weight \(-1\). Therefore, the only possibilities are:
\begin{enumerate}
    \item The spin biclique has weight 0, yielding the empty graph \( \overline{K_n} \).
    \item The spin biclique is \( SB_{n,0} \) with weight 1, yielding the complete graph \( K_n \).
\end{enumerate}
\end{proof}

Next, we characterize graphs with \( gc(G) = 2 \) and \( gc(G) = 3 \). The proofs of the following two theorems are deferred to Appendix~\ref{sec:appen-proofs}.

\begin{theorem}
Let \( G \) be a simple unweighted graph. Then
\[
gc(G) = 2 \iff \exists\, a, b \in \mathbb{N},\, (a, b) \ne (1, 1) \text{ such that } G = K_a \cup K_b \text{ or } G = K_{a,b}.
\]
\label{thm:gc=2}
\end{theorem}

\begin{theorem}
There exists no simple unweighted graph \( G \) with \( gc(G) = 3 \).
\label{thm:gc!=3}
\end{theorem}

Analyzing graphs with small graph coupling numbers suggests that a graph and its complement may share the same graph coupling number. Although we were unable to prove this conjecture in full generality, we establish a slightly weaker result in the following section.

\subsection{Fundamental Theoretical Results \label{sec:fundam}}

This section establishes key relationships between the graph coupling number \( gc \) of a graph and those of its spanning subgraphs, its complement, and its induced subgraphs.

We begin by generalizing Corollary~2 of~\cite{rajakumar2022generating}, which concerns the graph coupling number of edge-disjoint spanning subgraphs.

\begin{theorem}
Let \( G_1 = (V, E_1) \) and \( G_2 = (V, E_2) \) be simple undirected graphs on the same vertex set \( V \), with disjoint edge sets \( E_1 \cap E_2 = \emptyset \). Let \( G = (V, E_1 \cup E_2) \) denote their union. Then:
\[
\left| gc(G_1) - gc(G_2) \right| \le gc(G) \le gc(G_1) + gc(G_2).
\]
\label{thm:union}
\end{theorem}

\begin{proof}
The upper bound \( gc(G) \le gc(G_1) + gc(G_2) \) was shown in~\cite{rajakumar2022generating}. To prove the lower bound, let \( \mathbf{A}_1 \), \( \mathbf{A}_2 \), and \( \mathbf{A} \) denote the adjacency matrices of \( G_1 \), \( G_2 \), and \( G \), respectively, so that \( \mathbf{A}_1 = \mathbf{A} - \mathbf{A}_2 \).

Let \( (\mathbf{P}, \mathbf{W}) \) be an optimal solution for \( G \), and let \( (\mathbf{P}_2, \mathbf{W}_2) \) be an optimal solution for \( G_2 \). Define:
\[
\mathbf{P}' = \begin{bmatrix} \mathbf{P} \\ \mathbf{P}_2 \end{bmatrix}, \qquad
\mathbf{W}' = \begin{bmatrix} \mathbf{W} & \mathbf{0} \\ \mathbf{0} & -\mathbf{W}_2 \end{bmatrix}.
\]

Then \( (\mathbf{P}', \mathbf{W}') \) is a feasible solution for \( G_1 \), since:
\[
\mathbf{A}_1 = \mathbf{A} - \mathbf{A}_2 = \mathbf{P}^\top \mathbf{W} \mathbf{P} \odot \mathbf{J} - \mathbf{P}_2^\top \mathbf{W}_2 \mathbf{P}_2 \odot \mathbf{J} = \mathbf{P}'^\top \mathbf{W}' \mathbf{P}' \odot \mathbf{J}.
\]

Since \( (\mathbf{P}', \mathbf{W}') \) may not be optimal, we conclude:
\[
gc(G_1) \le gc(G) + gc(G_2) \quad \Rightarrow \quad gc(G_1) - gc(G_2) \le gc(G).
\]
By symmetry, interchanging the roles of \( G_1 \) and \( G_2 \) yields the reverse inequality:
\[
gc(G_2) - gc(G_1) \le gc(G),
\]
so the result follows:
\[
\left| gc(G_1) - gc(G_2) \right| \le gc(G).
\]
\end{proof}

As a consequence, the graph coupling number of a graph and that of its complement differ by at most one.

\begin{corollary}
Let \( H \) be a simple undirected graph, and let \( \overline{H} \) denote its complement. Then:
\[
\left| gc(H) - gc(\overline{H}) \right| \le 1.
\]
\end{corollary}

\begin{proof}
Apply Theorem~\ref{thm:union} with \( G_1 = H \), \( G_2 = \overline{H} \), and \( G = K_n \), the complete graph on \( n = |V(H)| \) vertices. By Theorem~\ref{thm:gc=1}, \( gc(K_n) = 1 \), so:
\[
\left| gc(H) - gc(\overline{H}) \right| \le gc(K_n) = 1.
\]
\end{proof}

We next consider the behavior of the graph coupling number under vertex restriction.

\begin{theorem}
Let \( S \subseteq V(G) \), and let \( G' = G[S] \) be the subgraph induced on \( S \). Then:
\[
gc(G') \le gc(G).
\]
\label{thm:induced}
\end{theorem}

\begin{proof}
Let \( (\mathbf{P}, \mathbf{W}) \) be an optimal solution for \( G \), so that \( \mathbf{A} = \mathbf{P}^\top \mathbf{W} \mathbf{P} \odot \mathbf{J} \) is the adjacency matrix of \( G \). From Equation~\ref{eq:column-wise-dot}, for all \( i, j \in S \), \( i < j \), we have:
\[
\mathbf{A}_{i,j} = w^\top (\mathbf{P}_i \odot \mathbf{P}_j),
\]
where \( w \in \mathbb{R}^k \) is the vector of diagonal entries of \( \mathbf{W} \), and \( \mathbf{P}_i \) denotes the \( i \)-th column of \( \mathbf{P} \).

Let \( \mathbf{P}' \in \{\pm 1\}^{k \times |S|} \) denote the submatrix of \( \mathbf{P} \) restricted to columns indexed by vertices in \( S \), and let \( \mathbf{A}' \) be the adjacency matrix of \( G' \). Then:
\[
\mathbf{A}'_{i,j} = w^\top (\mathbf{P}'_i \odot \mathbf{P}'_j) \quad \forall\, i < j.
\]
Hence:
\[
\mathbf{A}' = \mathbf{P}'^\top \mathbf{W} \mathbf{P}' \odot \mathbf{J},
\]
so \( (\mathbf{P}', \mathbf{W}) \) is a feasible solution for \( G' \), which implies \( gc(G') \le gc(G) \).
\end{proof}

\subsection{Lower Bound \label{sec:lower}}

To the best of our knowledge, no prior work has established a lower bound on the graph coupling number \( gc \) for unweighted graphs. In what follows, we propose such a bound, based on classical results concerning the rank of matrix products.

We begin by presenting an equivalent formulation of the feasibility condition in the Graph Coupling Problem.

\begin{theorem}
Let \( \mathbf{P} \in \{-1, +1\}^{k \times n} \) and let \( \mathbf{W} \in \mathbb{R}^{k \times k} \) be a diagonal matrix. Then:
\[
\mathbf{A} = \mathbf{P}^\top \mathbf{W} \mathbf{P} \odot \mathbf{J} \iff \mathbf{A} + \operatorname{tr}(\mathbf{W}) \mathbf{I} = \mathbf{P}^\top \mathbf{W} \mathbf{P}.
\]
\label{thm:hadel}
\end{theorem}

\begin{proof}
\(\Longleftarrow:\) Immediate, since the Hadamard product with \( \mathbf{J} \) zeroes out the diagonal of \( \mathbf{P}^\top \mathbf{W} \mathbf{P} \).

\(\Longrightarrow:\) Assume \( \mathbf{A} = \mathbf{P}^\top \mathbf{W} \mathbf{P} \odot \mathbf{J} \). Then the \( (i,i) \)-th entry of \( \mathbf{P}^\top \mathbf{W} \mathbf{P} \) is
\[
w^\top (\mathbf{P}_i \odot \mathbf{P}_i) = w^\top \mathds{1} = \operatorname{tr}(\mathbf{W}),
\]
where \( w \) is the vector of diagonal entries of \( \mathbf{W} \). For \( i \neq j \), the \( (i,j) \)-th entry is equal to \( \mathbf{A}_{i,j} \), since \( \mathbf{J}_{i,j} = 1 \). Hence, the identity \( \mathbf{P}^\top \mathbf{W} \mathbf{P} = \mathbf{A} + \operatorname{tr}(\mathbf{W}) \mathbf{I} \) follows.
\end{proof}

Next, we invoke a classical result from linear algebra on the rank of a product of two matrices~\cite{strang2023linear}:
\[
\operatorname{rank}(\mathbf{A}\mathbf{B}) \le \min\{\operatorname{rank}(\mathbf{A}), \operatorname{rank}(\mathbf{B})\}.
\]
Which can be easily generalized to 3 matrices by applying the same inequality twice: 
\begin{equation}
    \operatorname{rank}(\mathbf{A} \mathbf{B} \mathbf{C}) \le \min\{\operatorname{rank}(\mathbf{A}), \operatorname{rank}(\mathbf{B}), \operatorname{rank}(\mathbf{C})\}.
    \label{eq:rank3}
\end{equation}

Applying the inequality~\eqref{eq:rank3} to the feasibility condition from Theorem~\ref{thm:hadel} yields:
\[
\operatorname{rank}(\mathbf{A} + \operatorname{tr}(\mathbf{W}) \mathbf{I}) = \operatorname{rank}(\mathbf{P}^\top \mathbf{W} \mathbf{P}) \le \operatorname{rank}(\mathbf{W}) = \|\mathbf{W}\|_0.
\]
 Note that the identity \(\operatorname{rank}(\mathbf{W}) = \|\mathbf{W}\|_0\) holds because the number of non-zero entries of a diagonal matrix is equal to its rank. This inequality provides a necessary condition for feasibility. If we can lower-bound the rank of \( \mathbf{A} + tr(\mathbf{W}) \mathbf{I} \), then we obtain a valid lower bound on the graph coupling number, since \( \|\mathbf{W}\|_0 \) is the objective.

We are now ready to state our main result on lower-bound.

\begin{theorem}
Let \( \mathbf{A} \) be the adjacency matrix of a simple undirected graph \( G \). Then:
\[
\min_{\alpha \in \mathbb{R}} \operatorname{rank}(\mathbf{A} + \alpha \mathbf{I}) \le gc(G).
\]
\label{thm:lower}
\end{theorem}

\begin{proof}
For any feasible pair \( (\mathbf{P}, \mathbf{W}) \), we have:
\[
\min_{\alpha \in \mathbb{R}} \operatorname{rank}(\mathbf{A} + \alpha \mathbf{I}) \le \operatorname{rank}(\mathbf{A} + \operatorname{tr}(\mathbf{W}) \mathbf{I}) = \operatorname{rank}(\mathbf{P}^\top \mathbf{W} \mathbf{P}) \le \|\mathbf{W}\|_0.
\]
Taking \( \mathbf{W} \) to be an optimal solution yields the result.
\end{proof}

We observe that \( \operatorname{rank}(\mathbf{A} + \alpha \mathbf{I}) \) is closely related to the eigenvalue structure of \( \mathbf{A} \). This motivates the following reformulation:

\begin{corollary}
Let \( G \) be a simple graph with \( n \) vertices and adjacency matrix \( \mathbf{A} \). Then:
\[
n - \max_{\lambda \in \operatorname{spec}(\mathbf{A})} \operatorname{mult}(\lambda) \le gc(G),
\]
where \( \operatorname{mult}(\lambda) \) denotes the algebraic multiplicity of eigenvalue \( \lambda \).
\label{cor:geommult}
\end{corollary}

\begin{proof}
From Theorem~\ref{thm:lower}, we have:
\[
\min_{\alpha \in \mathbb{R}} \operatorname{rank}(\mathbf{A} + \alpha \mathbf{I}) \le gc(G).
\]
Consider the rank of \( \mathbf{A} + \alpha \mathbf{I} \) for different values of \( \alpha \):
\begin{itemize}
    \item If \( -\alpha \) is not an eigenvalue of \( \mathbf{A} \), then \( \mathbf{A} + \alpha \mathbf{I} \) is nonsingular and has full rank \( n \).
    \item If \( -\alpha \) is an eigenvalue of \( \mathbf{A} \), then the nullity of \( \mathbf{A} + \alpha \mathbf{I} \) equals the multiplicity of \( -\alpha \). Since \( \mathbf{A} \) is symmetric, it is diagonalizable, and thus its geometric and algebraic multiplicities coincide. Hence,
    \[
    \operatorname{rank}(\mathbf{A} + \alpha \mathbf{I}) = n - \operatorname{mult}(-\alpha).
    \]
\end{itemize}
Therefore,
\[
\min_{\alpha \in \mathbb{R}} \operatorname{rank}(\mathbf{A} + \alpha \mathbf{I}) = n - \max_{\lambda \in \operatorname{spec}(\mathbf{A})} \operatorname{mult}(\lambda),
\]
which completes the proof.
\end{proof}

\begin{remark}
The lower bounds in Theorem~\ref{thm:lower} and Corollary~\ref{cor:geommult} hold for both weighted and unweighted simple undirected graphs.
\end{remark}

\subsection{Lower Bound for Specific Graph Families \label{sec:lower_specific}}

We now apply Theorem~\ref{thm:lower} and Corollary~\ref{cor:geommult} to establish lower bounds on the graph coupling number \( gc \) for specific families of graphs. We begin with cliques.

\begin{theorem}
For all \( q \le \frac{n}{2} \), we have:
\[
gc(K_q \cup \overline{K_{n - q}}) \ge q.
\]
\end{theorem}

\begin{proof}
Without loss of generality, we may assume \( q \ge 2 \), since if \( q = 1 \), the inequality holds trivially.

According to \cite[pg.~8]{brouwer2011spectra}, the eigenvalues of the complete graph \( K_q \) are \( -1 \) and \( q - 1 \). We consider the matrix \( \mathbf{A} + \alpha \mathbf{I} \) in the context of Theorem~\ref{thm:lower}. We analyze three cases:
\begin{itemize}
    \item If \( \alpha = 1 \), then \( \operatorname{rank}(\mathbf{A} + \alpha \mathbf{I}) \ge n - q \ge q \), since the last \( n - q \) rows of \( \mathbf{A} + \alpha \mathbf{I} \) are linearly independent.
    \item If \( \alpha = -(q - 1) \), the same argument applies.
    \item If \( \alpha \notin \{1, -(q - 1)\} \), then the principal \( q \times q \) submatrix corresponding to the clique \( K_q \) is full rank, so \( \operatorname{rank}(\mathbf{A} + \alpha \mathbf{I}) \ge q \).
\end{itemize}

Thus, in all three cases, \( \operatorname{rank}(\mathbf{A} + \alpha \mathbf{I}) \ge q \). By Theorem~\ref{thm:lower}, it follows that:
\[
q \le \min_{\alpha \in \mathbb{R}} \operatorname{rank}(\mathbf{A} + \alpha \mathbf{I}) \le gc(G).
\]
\end{proof}

We next consider perfect matching graphs.

\begin{corollary}
For all \( q \in \mathbb{N} \), we have:
\[
gc(PM_q) \ge q.
\]
\label{cor:pmlower}
\end{corollary}

\begin{proof}
Let \( \mathbf{A} \) be the adjacency matrix of \( PM_q \). We apply Theorem~\ref{thm:lower} to the matrix \( \mathbf{A} + \alpha \mathbf{I} \), which has the block-diagonal form:
\[
\mathbf{A} + \alpha \mathbf{I} =
\begin{bmatrix}
    \alpha & 1     &        &        &        \\
    1      & \alpha&        &        &        \\
           &       & \ddots &        &        \\
           &       &        & \alpha & 1      \\
           &       &        & 1      & \alpha \\
\end{bmatrix}
\]

There are \( q \) such \( 2 \times 2 \) blocks. When \( \alpha = \pm 1 \), each block has rank 1, so the total rank is \( q \). For any other \( \alpha \in \mathbb{R} \), each block has rank 2, and the total rank is \( 2q \). Thus,
\[
\min_{\alpha \in \mathbb{R}} \operatorname{rank}(\mathbf{A} + \alpha \mathbf{I}) = q \Longrightarrow gc(PM_q) \ge q.
\]
\end{proof}

Next, we consider path graphs.

\begin{corollary}
For all \( n \in \mathbb{N} \), we have:
\[
gc(P_n) \ge n - 1.
\]
\label{cor:pathlower}
\end{corollary}

\begin{proof}
According to \cite[pg.~9]{brouwer2011spectra}, the eigenvalues of the adjacency matrix \( \mathbf{A} \) of \( P_n \) are
\[
2 \cos\left(\frac{k\pi}{n + 1}\right), \quad \text{for } k \in [n],
\]
all of which are distinct. Thus, the maximum multiplicity of any eigenvalue is 1. Applying Corollary~\ref{cor:geommult} yields:
\[
gc(P_n) \ge n - 1.
\]
\end{proof}

The result in Corollary~\ref{cor:pathlower} is significant: it shows that for every positive integer \( n \), the path graph \( P_n \) satisfies \( gc(P_n) \ge n - 1 \). This confirms the \emph{order-optimality} of the Union of Stars construction for general unweighted graphs proposed in~\cite{rajakumar2022generating}, which was previously only conjectured.

We conclude with a lower bound for cycles.

\begin{corollary}
For all \( n \in \mathbb{N} \), we have:
\[
gc(C_n) \ge n - 2.
\]
\end{corollary}

\begin{proof}
Since \( P_{n-1} \) is an induced subgraph of \( C_n \), Theorem~\ref{thm:induced} gives:
\[
gc(C_n) \ge gc(P_{n - 1}) \ge n - 2.
\]
\end{proof}

\subsection{Upper Bound for Specific Graph Families \label{sec:upperspec}}

Solving the Graph Coupling Problem for special families of graphs is not only an interesting problem in its own right but also, as we demonstrate later, useful for improving the upper bound of \( gc \) for general unweighted graphs. Here, we discuss feasible solutions for cliques, unions of vertex-disjoint cliques, perfect matchings, cycles, and paths. Although these solutions are not necessarily optimal, they are close, as the lower bounds on \( gc \) for such graphs in the previous section demonstrate.

We begin by showing that the graph coupling number of a clique on \( q \) vertices is at most \( q + 2 \).

\begin{theorem}
    \( gc(K_q \cup \overline{K_{n - q}}) \le q + 2 \) for all \( q \in [n] \).
    \label{thm:clique}
\end{theorem}

\begin{proof}
    One can decompose \( G = K_q \cup \overline{K_{n - q}} \) into the following \( q + 2 \) spin bicliques:
    \[
    G = - \sum_{i=1}^{q} \frac{1}{4} SB_{\{i\}, V \setminus \{i\}} + \frac{1}{4} SB_{\{1, 2, \dots, q\}, V \setminus \{1, 2, \dots, q\}} + \frac{q-1}{4} K_n.
    \]
    Therefore, \( gc(G) \le q + 2 \).
\end{proof}

We can strengthen this result by using Theorem~\ref{thm:union}:

\begin{corollary}
    \( gc(K_q \cup \overline{K_{n - q}}) \le \min\{q, n - q\} + 2 \) for all \( q \in [n] \).
\end{corollary}

\begin{proof}
    We consider two cases:
    \begin{itemize}
        \item If \( q \le n - q \), we are done by Theorem~\ref{thm:clique}.
        \item If \( q > n - q \), let \( G_1 = K_q \cup \overline{K_{n - q}} \), and \( G_2 = \overline{K_q} \cup K_{n - q} \). Then Theorem~\ref{thm:union} gives:
        \[
        |gc(G_1) - gc(G_2)| \le 2.
        \]
        Since \( gc(K_q \cup K_{n - q}) = 2 \) by Theorem~\ref{thm:gc=2} with corresponding \(\mathbf{P}\) having the two rows: \(r_{\{1, 2, \dots, q \}}\) and \(\mathds{1}\). With the procedure introduced in the proof of Theorem~\ref{thm:union}, one can construct a feasible pair \( (\mathbf{P}_1, \mathbf{W}_1) \) for \( G_1 \) from \( (\mathbf{P}_2, \mathbf{W}_2) \) for \( G_2 \) by simply adding the rows \( r_{\{1, 2, \dots, q\}} \) and \( \mathds{1}_n \) to the rows of \( \mathbf{P}_2 \), and appending the corresponding weights to \( \mathbf{W}_2 \). Since \(\mathbf{P}_2\), constructed using Theorem~\ref{thm:clique}, already includes these rows, we conclude that there is a feasible solution for \( G_1 \) with at most \( n - q + 2 \) rows.
    \end{itemize}
\end{proof}

Next, we give a construction for the union of vertex-disjoint cliques.

\begin{theorem}
    Let \(S_1, S_2, \dots, S_q\) be a partition of the vertices of the graph \(G\) where each set of vertices \(S_i\) form a clique. i.e. \(G = K_{S_1} \cup K_{S_2} \cup \dots \cup K_{S_q}\) then: \[ gc(G) \le q + 1 \].
    \label{thm:undiscli}
\end{theorem}

\begin{proof}
    Let \( G = K_{S_1} \cup K_{S_2} \cup \dots \cup K_{S_q} \), and let \( G_i = SB_{S_i, V \setminus S_i} \). Then:
    \[
    G = \left(1 - \frac{q}{4}\right) K_V + \sum_{i = 1}^{q} \frac{1}{4} G_i \quad \Longrightarrow \quad gc(G) \le q + 1.
    \]
\end{proof}

A perfect matching graph is a special case of a union of vertex-disjoint cliques, where all cliques have size 2.

\begin{corollary}
    \( gc(PM_q) \le q + 1 \) for all \( q \in \mathbb{N} \).
    \label{cor:pmupper}
\end{corollary}

\begin{proof}
    Immediate from Theorem~\ref{thm:undiscli}.
\end{proof}

Next, we study cycles. An even cycle can be decomposed into two perfect matchings by taking alternating edges. For odd cycles, a more detailed construction is needed. The following corollary summarizes our results.

\begin{corollary}
    \( gc(C_n) \le n + 1 \).
    \label{cor:cycleupper}
\end{corollary}

\begin{proof}
    We consider two cases:
    \begin{enumerate}
        \item \( n \) even: \\
        The edges of an even cycle can be partitioned into two perfect matchings using odd- and even-indexed edges. Since each perfect matching can be constructed using \( \frac{n}{2} + 1 \) spin bicliques (one of which is \( K_n \)), the entire cycle can be constructed using \( n + 1 \) spin bicliques. Note that we can combine the \( K_n \) terms by summing their weights.
        \item \( n \) odd: \\
        Enumerate the vertices and edges of the cycle as \( v_1, \dots, v_n \) and \( e_1, \dots, e_n \). Let \( G_1 \) consist of edges \( e_1, e_3, \dots, e_{n-2} \), \( G_2 \) of \( e_2, e_4, \dots, e_{n-1} \), and \( G_3 \) of the single edge \( e_n \). Then \( C_n = G_1 + G_2 + G_3 \) is an edge partition.

        \( G_1 \) and \( G_2 \) are each of the form described in Theorem~\ref{thm:undiscli}, so we use that construction:
        \[
        G_1 = \left(1 - \frac{n + 1}{8} \right) K_V + \frac{1}{4} SB_{\{v_n\}, V \setminus \{v_n\}} + \frac{1}{4} \sum_{i = 1}^{\frac{n - 1}{2}} SB_{\{v_{2i - 1}, v_{2i}\}, V \setminus \{v_{2i - 1}, v_{2i}\}},
        \]
        \[
        G_2 = \left(1 - \frac{n + 1}{8} \right) K_V + \frac{1}{4} SB_{\{v_1\}, V \setminus \{v_1\}} + \frac{1}{4} \sum_{i = 1}^{\frac{n - 1}{2}} SB_{\{v_{2i}, v_{2i + 1}\}, V \setminus \{v_{2i}, v_{2i + 1}\}}.
        \]
        Using Theorem~\ref{thm:clique}, we have:
        \[
        G_3 = \frac{1}{4} K_V + \frac{1}{4} SB_{\{v_1, v_n\}, V \setminus \{v_1, v_n\}} - \frac{1}{4} SB_{\{v_1\}, V \setminus \{v_1\}} - \frac{1}{4} SB_{\{v_n\}, V \setminus \{v_n\}}.
        \]
        Putting everything together:
        \[
        C_n = G_1 + G_2 + G_3 = \left(2 - \frac{n}{4} \right) K_V + \frac{1}{4} \sum_{i=1}^{n - 1} SB_{\{v_i, v_{i+1}\}, V \setminus \{v_i, v_{i+1}\}} + \frac{1}{4} SB_{\{v_1, v_n\}, V \setminus \{v_1, v_n\}}.
        \]
        Therefore, \( gc(C_n) \le n + 1 \).
    \end{enumerate}
\end{proof}

The path graph \( P_n \) is an induced subgraph of the cycle \( C_{n+1} \), so we can use Theorem~\ref{thm:induced} to obtain the following upper bound.

\begin{corollary}
    \( gc(P_n) \le n + 2 \).
\end{corollary}

\begin{proof}
    \[
    gc(P_n) \le gc(C_{n+1}) \le n + 2.
    \]
    The first inequality holds because \( P_n \) is an induced subgraph of \( C_{n+1} \) by Theorem~\ref{thm:induced}, and the second follows from Corollary~\ref{cor:cycleupper}.
\end{proof}

Combining the upper and lower bounds derived for these specific graph families, we obtain:

\begin{itemize}
    \item \( q \le gc(K_q \cup \overline{K_{n - q}}) \le q + 2 \) for all \( q \le \frac{n}{2} \). Hence, the possible values for \( gc \) of such a clique on q vertices are \( q, q + 1, \) or \( q + 2 \).
    \item \( q \le gc(PM_q) \le q + 1 \), so the coupling number for perfect matchings is either \( q \) or \( q + 1 \).
    \item \( n - 1 \le gc(P_n) \le n + 2 \), giving four possible values for the coupling number of a path.
    \item \( n - 2 \le gc(C_n) \le n + 1 \), again yielding four possibilities for the coupling number of a cycle.
\end{itemize}

In particular, the fact that \( gc(PM_q) \in \{q, q + 1\} \) motivates further investigation of perfect matching graphs. In the following section, we explore which values of \( q \) lead to \( gc = q \) and which yield \( gc = q + 1 \).

\subsection{Connection to the Hadamard Conjecture \label{sec:hadamard}}

So far, we have established that \( q \le gc(PM_q) \le q + 1 \) from Corollary~\ref{cor:pmupper} and Theorem~\ref{cor:pmlower}. This narrows the graph coupling number of perfect matching graphs to just two possibilities. We argue that whenever a Hadamard matrix of order \( q \) exists, then \( gc(PM_q) = q \); otherwise, \( gc(PM_q) = q + 1 \). The existence of Hadamard matrices for all such orders, however, remains an open problem. We begin by formally defining Hadamard matrices.

\begin{Def}
An order-\( n \) \textbf{Hadamard matrix} is an \( n \times n \) matrix with entries in \( \{-1, +1\} \), whose rows are pairwise orthogonal. Examples include:
\[
\mathbf{H}_1 = 
\begin{bmatrix}
    1
\end{bmatrix}, \quad
\mathbf{H}_2 = 
\begin{bmatrix}
    1 & 1 \\
    1 & -1
\end{bmatrix}, \quad
\mathbf{H}_4 = 
\begin{bmatrix}
    1 & 1 & 1 & 1 \\
    1 & -1 & 1 & -1 \\
    1 & 1 & -1 & -1 \\
    1 & -1 & -1 & 1
\end{bmatrix}.
\]
\end{Def}

Equivalently, a matrix \( \mathbf{H} \in \{-1, +1\}^{n \times n} \) is Hadamard if and only if \( \mathbf{H} \mathbf{H}^\top = n \mathbf{I} \). The orthogonality of the columns is also a necessary and sufficient condition for \( \mathbf{H} \) to be Hadamard.

From the definition, it is clear that no Hadamard matrix of odd order exists: the dot product of any two vectors \( \mathbf{v}_1, \mathbf{v}_2 \in \{-1, +1\}^{2k+1} \) cannot be zero. Less obvious, but still provable, is that no Hadamard matrix of order \( 4k + 2 \) exists for any \( k \in \mathbb{N} \). 

To see this, suppose for contradiction that a Hadamard matrix \( \mathbf{H} \) of order \( 4k + 2 \) exists. Let \( \mathbf{v}_1, \mathbf{v}_2, \mathbf{v}_3 \) be three distinct rows of \( \mathbf{H} \). Since \( \mathbf{v}_1 \cdot \mathbf{v}_2 = 0 \), it must be the case that \( 2k + 1 \) entries differ between them. The same must hold between \( \mathbf{v}_2 \) and \( \mathbf{v}_3 \). Then \( \mathbf{v}_3 \) differs from \( \mathbf{v}_1 \) in at most \( 4k + 2 - 2C \) entries, where \( C \) is the number of common flipped positions. Therefore, the number of differing entries between \( \mathbf{v}_1 \) and \( \mathbf{v}_3 \) is even, contradicting the orthogonality condition, which would require \( 2k + 1 \) (odd) differences. Thus, there cannot be three mutually orthogonal \( \pm 1 \)-vectors in \( \mathbb{R}^{4k + 2} \), let alone \( 4k + 2 \) of them.

Consequently, no Hadamard matrix exists for any \( n \ge 4 \) such that \( n \equiv 1, 2, 3 \pmod{4} \). All existing Hadamard matrices must be of order 1, 2, or a multiple of 4~\cite{hedayat1978hadamard}.

\begin{conjecture}[Hadamard Conjecture]
For every integer \( n \ge 4 \) such that \( n \equiv 0 \pmod{4} \), there exists an order-\( n \) Hadamard matrix.
\end{conjecture}

The conjecture remains unresolved. Although we do not aim to prove or disprove it here, we demonstrate a compelling equivalence with the Graph Coupling Problem. From Corollary~\ref{cor:pmupper} and Theorem~\ref{cor:pmlower}, we know:
\[
q \le gc(PM_q) \le q + 1.
\]
The following theorem provides an exact characterization.

\begin{theorem}
\[
gc(PM_q) = q \iff \text{there exists a Hadamard matrix of order } q.
\]
\end{theorem}

\begin{proof}
Let \( \mathbf{A} \) be the adjacency matrix of \( PM_q \).

\noindent\textbf{(\( \Leftarrow \))}: Suppose there exists a Hadamard matrix \( \mathbf{H} \) of order \( q \). Construct a matrix \( \mathbf{P} \in \{-1, +1\}^{q \times 2q} \) by duplicating each column of \( \mathbf{H} \) in place. That is, columns 1 and 2 of \( \mathbf{P} \) equal the first column of \( \mathbf{H} \); columns 3 and 4 equal the second column, and so on. Then:
\[
\mathbf{P}^\top \left(\frac{1}{q} \mathbf{I}\right) \mathbf{P} = \mathbf{A} + \mathbf{I},
\]
so \( (\mathbf{P}, \mathbf{W} = \frac{1}{q} \mathbf{I}) \) is a feasible solution. Combining with Theorem~\ref{cor:pmlower}, we conclude \( gc(PM_q) = q \).

\noindent\textbf{(\( \Rightarrow \))}: Suppose \( gc(PM_q) = q \), so there exists a feasible pair \( (\mathbf{P}, \mathbf{W}) \) with \( \mathbf{P} \in \mathbb{R}^{q \times 2q} \), \( \mathbf{W} \in \mathbb{R}^{q \times q} \) diagonal, such that:
\[
\mathbf{P}^\top \mathbf{W} \mathbf{P} = \mathbf{A} + \operatorname{tr}(\mathbf{W}) \mathbf{I}.
\]
Since the left-hand side has rank at most \( q \), it must be that \( \operatorname{tr}(\mathbf{W}) \in \{-1, 1\} \). Extract the odd-numbered columns of \( \mathbf{P} \) to form matrix \( \mathbf{H} \in \mathbb{R}^{q \times q} \). Then:
\begin{itemize}
    \item If \( \operatorname{tr}(\mathbf{W}) = 1 \), then \( \mathbf{H}^\top \mathbf{W} \mathbf{H} = \mathbf{I} \), so \( \mathbf{W}^{-1} = \mathbf{H} \mathbf{H}^\top \).
    \item If \( \operatorname{tr}(\mathbf{W}) = -1 \), then \( \mathbf{H}^\top \mathbf{W} \mathbf{H} = -\mathbf{I} \), and similarly \( \mathbf{W}^{-1} = -\mathbf{H} \mathbf{H}^\top \).
\end{itemize}

In both cases, since \( \mathbf{W} \) is diagonal, \( \mathbf{H} \mathbf{H}^\top \) must also be diagonal, implying that the rows of \( \mathbf{H} \) are orthogonal. Thus, \( \mathbf{H} \) is a Hadamard matrix of order \( q \).
\end{proof}

\subsection{Upper Bound for \texorpdfstring{$gc$}{gc} of General Unweighted Graphs \label{sec:generupper}}

In Section~\ref{sec:upperspec}, we discussed feasible solutions to the Graph Coupling Problem for specific families of graphs. We now consider the more general case and describe how to construct a feasible pair \( (\mathbf{P}, \mathbf{W}) \) for any unweighted graph.

Rajakumar et al.~\cite{rajakumar2022generating} proposed an algorithm that decomposes any unweighted graph into a union of at most \( n - 1 \) edge-disjoint stars, where \( n \) is the number of vertices. Their construction stacks the rows required for each star in the matrix \( \mathbf{P} \), resulting in an upper bound of \( 3n - 2 \) for \( gc(G) \), since each star requires three rows, plus a common row of all ones. This method is referred to as the \emph{Union of Stars} algorithm.

We propose a new construction that improves on this method by processing two stars at a time—a structure we refer to as a \emph{double-star}. To motivate our approach, consider the graph in Figure~\ref{fig:twostars}, called the \emph{prototype graph}, which is the simplest instance of a double-star with centers at vertices~1 and~2.

\begin{lemma}
    The graph coupling number of the prototype graph is \( gc(G) = 6 \), and the following pair \( (\mathbf{P}, \mathbf{W}) \) is optimal:
    \[
    \mathbf{P} =
    \begin{bmatrix}
        1 & 1 & 1 & 1 & 1 & 1 \\
        1 & 1 & 1 & 1 & -1 & -1 \\
        1 & 1 & 1 & -1 & -1 & 1 \\
        1 & 1 & -1 & -1 & -1 & -1 \\
        1 & -1 & 1 & -1 & -1 & 1 \\
        1 & -1 & -1 & -1 & 1 & 1 \\
    \end{bmatrix}, \quad
    \mathbf{W} =
    \begin{bmatrix}
        \frac{1}{4} & 0 & 0 & 0 & 0 & 0 \\
        0 & \frac{1}{4} & 0 & 0 & 0 & 0 \\
        0 & 0 & -\frac{1}{4} & 0 & 0 & 0 \\
        0 & 0 & 0 & -\frac{1}{4} & 0 & 0 \\
        0 & 0 & 0 & 0 & \frac{1}{4} & 0 \\
        0 & 0 & 0 & 0 & 0 & -\frac{1}{4} \\
    \end{bmatrix}.
    \]
    \label{lem:twostargc6}
\end{lemma}

\begin{proof}
    Substituting \( \mathbf{P} \) and \( \mathbf{W} \) into the problem formulation verifies that the pair is feasible. To show that six is the minimum number of rows, we perform an exhaustive search over all possible matrices \( \mathbf{P} \in \{\pm 1\}^{5 \times 6} \). Once \( \mathbf{P} \) is fixed, the condition \( \mathbf{P}^\top \mathbf{W} \mathbf{P} = \mathbf{A} + \operatorname{tr}(\mathbf{W}) \mathbf{I} \) becomes a linear system in the entries of \( \mathbf{W} \), which can be solved efficiently. No feasible solutions are found with fewer than six rows, confirming the optimality of the proposed pair. See Appendix~\ref{sec:appen-bruteforce} for the exhaustive search algorithm. 
\end{proof}

\begin{figure}[h]
    \centering
    \includegraphics[width=0.3\linewidth]{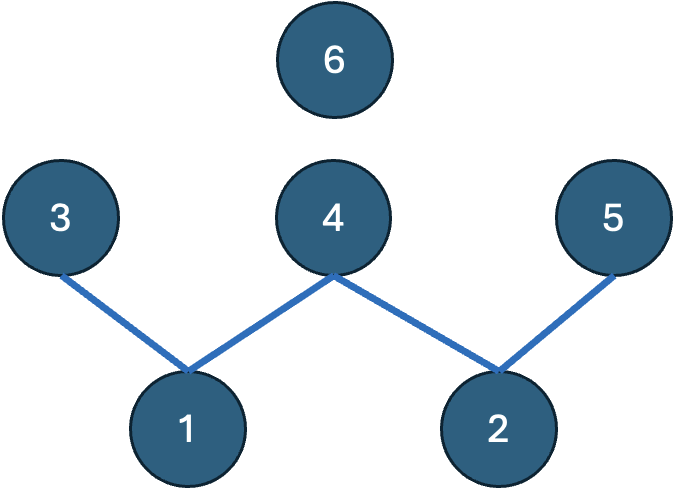}
    \caption{The prototype graph: two stars with centers at vertices~1 and~2, with \( gc = 6 \).}
    \label{fig:twostars}
\end{figure}

\begin{theorem}
Let \( G = (V, E) \) be a general double-star graph on vertex set \( V = \{v_1, v_2\} \cup V_3 \cup V_4 \cup V_5 \cup V_6 \), where the sets \( V_3, V_4, V_5, V_6 \) are pairwise disjoint and disjoint from \( \{v_1, v_2\} \). Let the edge set be
\[
E = \big\{ \{v_1, v\} : v \in V_3 \big\} \cup \big\{ \{x, y\} : x \in \{v_1, v_2\}, y \in V_4 \big\} \cup \big\{ \{v_2, v\} : v \in V_5 \big\}.
\]
Then \( gc(G) \le 6 \).
\label{thm:twostareliminate}
\end{theorem}

\begin{proof}
In Lemma~\ref{lem:twostargc6}, we established a feasible solution \( (\mathbf{P}, \mathbf{W}) \) with six rows for the prototype graph, one of which is the all-ones row.

To generalize this to \( G \), assign to each vertex in \( V_i \) the same column of \( \mathbf{P} \) used for vertex~\( i \) in the prototype, for \( i \in \{3, 4, 5, 6\} \). Assign \( \mathbf{P}_1 \) and \( \mathbf{P}_2 \) to \( v_1 \) and \( v_2 \), respectively. Reuse the same matrix \( \mathbf{W} \). This defines a feasible pair for \( G \).

By Corollary~\ref{cor:colclass}, vertices in the same column class behave identically with respect to other classes, and intra-class edge weights equal \( \operatorname{tr}(\mathbf{W}) = 0 \), matching the edge structure of \( G \). Therefore, the construction is feasible and \( gc(G) \le 6 \).
\end{proof}

Theorem~\ref{thm:twostareliminate} enables us to eliminate each double-star using only 5 additional rows, since the all-ones row is shared across all such structures. This reduces the row cost per star by approximately \( \frac{1}{2} \) compared to the Union of Stars algorithm, which uses 3 rows per star.

To apply this construction, we repeatedly find pairs \( v_1, v_2 \) that are not adjacent and build the corresponding double-star subgraph and eliminate its edges from the remaining graph. We proceed as long as the graph contains a non-edge. The process terminates when the residual graph becomes a clique, which can be handled efficiently using Theorem~\ref{thm:clique}, requiring only \( q + 2 \) rows for a clique on \( q \) vertices.

\begin{algorithm}[H]
\caption{Union of Double-Stars Construction}
\begin{algorithmic}[1]
\State \textbf{Input:} A simple, undirected graph \( G = (V, E) \) with adjacency matrix \( \mathbf{A} \).
\State Initialize \( S \gets V \); initialize empty matrices \( \mathbf{P} \), \( \mathbf{W} \).
\While{\( S \ne \emptyset \) and \( \exists v_1, v_2 \in S \) with \( \mathbf{A}_{v_1, v_2} = 0 \)}
    \State Construct the double-star subgraph \( G' \) centered at \( v_1, v_2 \), with feasible pair \( (\mathbf{P}', \mathbf{W}') \).
    \State Append \( \mathbf{P}' \) to \( \mathbf{P} \); extend \( \mathbf{W} \) with \( \mathbf{W}' \) along the diagonal.
    \State Remove edges of \( G' \): \( G \gets G - E(G') \).
    \State Update \( S \gets S \setminus \{v_1, v_2\} \).
\EndWhile
\If{\( S \ne \emptyset \)} \Comment{Remaining graph is a clique}
    \State Construct clique solution \( (\mathbf{P}', \mathbf{W}') \) from Theorem~\ref{thm:clique}.
    \State Append \( \mathbf{P}' \) to \( \mathbf{P} \); extend \( \mathbf{W} \) with \( \mathbf{W}' \) along the diagonal.
\EndIf
\State \textbf{Output:} A feasible pair \( (\mathbf{P}, \mathbf{W}) \) for \( G \).
\end{algorithmic}
\end{algorithm}

\begin{theorem}
    The Union of Double-Stars algorithm yields a feasible solution to the Graph Coupling Problem with objective value at most \( 2.5n + 2 \).
\end{theorem}

\begin{proof}
Feasibility follows immediately because the proposed double-stars and the final clique partition the edges of the original graph.

For the row count, suppose the loop runs \( x \le \frac{n}{2} \) times, eliminating \( 2x \) vertices. This leaves \( n - 2x \) vertices forming a clique.

Each double-star uses 5 rows, with one shared all-ones row, giving \( 5x + 1 \) rows in total from the loop. The clique construction requires at most \( (n - 2x) + 1 \) additional rows (excluding the all-ones row). Thus, the total is:
\[
5x + 1 + (n - 2x + 1) = n + 3x + 2 \le n + 1.5n + 2 = 2.5n + 2.
\]
\end{proof}

This improves the theoretical upper bound from \( 3n - 2 \) (Union of Stars) to \( 2.5n + 2 \), a savings of roughly \( \frac{1}{6} \).

We also evaluate both algorithms empirically. We generate 34 Erdős–Rényi graphs with edge probabilities sampled uniformly from \( [0, 1] \), with sizes ranging from 4 to 20 vertices (two graphs per size). We discard two edgeless graphs, leaving 32 graphs for evaluation.

In both methods, subgraphs are selected greedily to maximize edge elimination per iteration. Redundant rows are removed from the final matrix \( \mathbf{P} \). Figure~\ref{fig:union-stars-double} shows that the Union of Double-Stars consistently outperforms the Union of Stars in practice.

\begin{figure}[H]
    \centering
    \subfloat[Scatter Plot]{\includegraphics[width=0.45\textwidth]{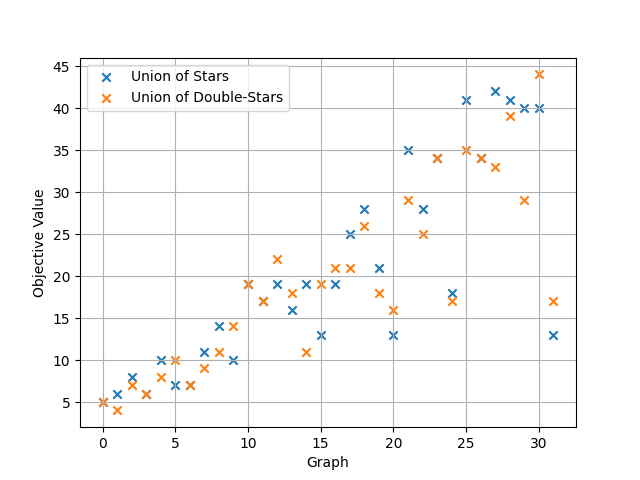}\label{fig:compare-scatter}}
    \hspace{0.05\textwidth}
    \subfloat[Box Plot]{\includegraphics[width=0.45\textwidth]{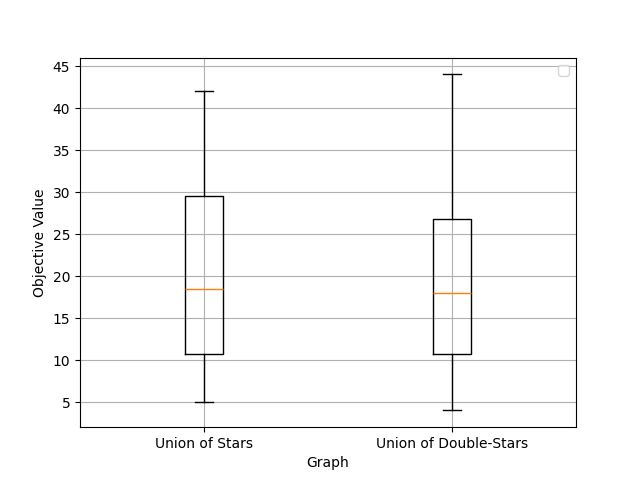}\label{fig:compare-box}}
    \caption{Comparison of Union of Stars and Union of Double-Stars on 32 Erdős–Rényi graphs.}
    \label{fig:union-stars-double}
\end{figure}

\subsection{CMIPGC \label{sec:comp}}

Finally, we propose a compact mixed-integer program for the Graph Coupling problem, which we refer to as \emph{CMIPGC}. The feasibility constraints are trilinear, as can be seen from the relation \( A_{i, j} = \mathbf{w}^\top (\mathbf{P}_i \odot \mathbf{P}_j) \) for all \( i < j \) from Equation \eqref{eq:column-wise-dot}. To formulate a mixed-integer linear program, we must linearize these constraints. The first step is to convert the \(-1\) and \(+1\) entries of \( \mathbf{P} \) into binary values in \( \{0, 1\} \). The following lemma establishes an equivalent feasibility condition in terms of the binary matrix \( \mathbf{P}' \).

\begin{lemma}
Let \( \mathbf{P} \in \{-1, 1\}^{k \times n} \), and let \( \mathbf{W} \) be a diagonal matrix. Define the binary matrix \( \mathbf{P}' = \frac{1}{2}(\mathbf{P} + \mathbf{1}) \), where \(-1\)'s are mapped to 0 and \(1\)'s stay the same. Assume that the first column of \( \mathbf{P} \), denoted \( \mathbf{P}_1 \), is \( \mathds{1} \). Then the first column of \( \mathbf{P}' \), denoted \( \mathbf{P}'_1 \), is also \( \mathds{1} \). Under this change of variables, the following equivalence holds:
\[
\mathbf{A} + \operatorname{tr}(\mathbf{W}) \mathbf{I} = \mathbf{P}^\top \mathbf{W} \mathbf{P}
\iff 
\mathbf{A} + \operatorname{tr}(\mathbf{W})
\begin{bmatrix}
4 & 2 & 2 & \dots & 2 \\
2 & 2 & 1 & \dots & 1 \\
2 & 1 & 2 & \dots & 1 \\
\vdots & \vdots & \vdots & \ddots & \vdots \\
2 & 1 & 1 & \dots & 2
\end{bmatrix}
+ 
\begin{bmatrix}
a_1 \\
a_1 \\
\vdots \\
a_1
\end{bmatrix}
+ 
\begin{bmatrix}
a_1^\top & a_1^\top & \dots & a_1^\top
\end{bmatrix}
= 4 \mathbf{P}'^\top \mathbf{W} \mathbf{P}'.
\]
\label{lem:binaryform}
\end{lemma}

\begin{proof}
Substitute \( \mathbf{P} = 2\mathbf{P}' - \mathbf{1} \) into the feasibility condition:
\begin{equation}
\begin{aligned}
\mathbf{A} + tr(\mathbf{W}) \mathbf{I} = \mathbf{P}^\top \mathbf{W} \mathbf{P} \iff 
\mathbf{A} + \operatorname{tr}(\mathbf{W}) \mathbf{I} 
&= (2\mathbf{P}' - \mathbf{1})^\top \mathbf{W} (2\mathbf{P}' - \mathbf{1}) \\
&= 4 \mathbf{P}'^\top \mathbf{W} \mathbf{P}' 
- 2 \cdot \mathbf{1}^\top \mathbf{W} \mathbf{P}' 
- 2 \cdot \mathbf{P}'^\top \mathbf{W} \mathbf{1} 
+ \mathbf{1}^\top \mathbf{W} \mathbf{1}.
\end{aligned}
\label{eq:binequiva}
\end{equation}

Next, we compute \( \mathbf{1}^\top \mathbf{W} \mathbf{P}' \) and its transpose for the right-hand side of Equation~\eqref{eq:binequiva}. Assume \( \mathbf{A} + \operatorname{tr}(\mathbf{W}) \mathbf{I} = \mathbf{P}^\top \mathbf{W} \mathbf{P} \). Then the first row of this identity is:
\[
a_1 + \operatorname{tr}(\mathbf{W}) e_1^\top = \mathds{1}^\top \mathbf{W} \mathbf{P},
\]
where \( \mathds{1} \) is the all-ones vector. Therefore,
\begin{equation}
\mathbf{1}^\top \mathbf{W} \mathbf{P}' 
= \frac{1}{2} \left( \mathbf{1}^\top \mathbf{W} \mathbf{P} + \operatorname{tr}(\mathbf{W}) \mathbf{1} \right)
= \frac{1}{2}
\left(
\begin{bmatrix}
a_1 + \operatorname{tr}(\mathbf{W}) e_1^\top \\
\vdots \\
a_1 + \operatorname{tr}(\mathbf{W}) e_1^\top
\end{bmatrix}
+ \operatorname{tr}(\mathbf{W}) \mathbf{1}
\right).
\label{eq:rowsumbin}
\end{equation}

Substituting Equation~\eqref{eq:rowsumbin} into Equation~\eqref{eq:binequiva} yields:
\begin{equation}
\begin{aligned}
\mathbf{A} + tr(\mathbf{W}) \mathbf{I} = \mathbf{P}^\top \mathbf{W} \mathbf{P} \iff \mathbf{A} + \operatorname{tr}(\mathbf{W}) \mathbf{I} 
&= 4 \mathbf{P}'^\top \mathbf{W} \mathbf{P}' \\
&\quad -
\left(
\begin{bmatrix}
a_1 + \operatorname{tr}(\mathbf{W}) e_1^\top \\
\vdots \\
a_1 + \operatorname{tr}(\mathbf{W}) e_1^\top
\end{bmatrix}
+ \operatorname{tr}(\mathbf{W}) \mathbf{1}
\right) \\
&\quad -
\left(
\begin{bmatrix}
a_1 + \operatorname{tr}(\mathbf{W}) e_1^\top \\
\vdots \\
a_1 + \operatorname{tr}(\mathbf{W}) e_1^\top
\end{bmatrix}^\top
+ \operatorname{tr}(\mathbf{W}) \mathbf{1}
\right) \\
&\quad + \operatorname{tr}(\mathbf{W}) \mathbf{1},
\end{aligned}
\label{eq:booleaneq}
\end{equation}
which simplifies to the desired result.
\end{proof}

Note that the assumption \( \mathbf{P}_1 = \mathds{1} \) is made without loss of generality. If a pair \( (\mathbf{P}, \mathbf{W}) \) satisfies the feasibility condition, then flipping the sign of any row of \( \mathbf{P} \) yields another feasible solution. This is because the spin bicliques \( SB_{V \setminus S, S} \) and \( SB_{S, V \setminus S} \) are identical. Therefore, by flipping the sign of each row of \( \mathbf{P} \) whose first entry is \(-1\), we obtain a matrix whose first column is equal to \( \mathds{1} \).

The feasibility condition from Lemma~\ref{lem:binaryform} remains trilinear when expressed in terms of the binary matrix \( \mathbf{P}' \). Specifically, the \((i,j)\)-th entry of \( \mathbf{P}'^\top \mathbf{W} \mathbf{P}' \) is given by a sum of trilinear terms:
\[
(\mathbf{P}'^\top \mathbf{W} \mathbf{P}')_{i,j} = \sum_{r=1}^k \mathbf{W}_{r,r} \mathbf{P}'_{r,i} \mathbf{P}'_{r,j} = \sum_{r=1}^k t_{i,j,r},
\]
where each term is defined as \( t_{i,j,r} = \mathbf{W}_{r,r} \mathbf{P}'_{r,i} \mathbf{P}'_{r,j} \).

To linearize the product \( \mathbf{P}'_{r,i} \mathbf{P}'_{r,j} \) of two binary variables, we introduce a new binary variable \( z_{i,j,r} \) and set \( z_{i,j,r} = \mathbf{P}'_{r,i} \mathbf{P}'_{r,j} \). This product is enforced through the constraints~\ref{eq:and1}--\ref{eq:and3}, as described by Padberg~\cite{padberg1989boolean}.

The resulting program now contains bilinear terms of the form \( t_{i,j,r} = \mathbf{W}_{r,r} z_{i,j,r} \), which we linearize using the standard big-\( M \) technique. Constraints~\ref{eq:binreal1}--\ref{eq:binreal2} ensure that \( t_{i,j,r} \) equals 0 when \( z_{i,j,r} = 0 \) and equals \( \mathbf{W}_{r,r} \) when \( z_{i,j,r} = 1 \).

Once the \( t_{i,j,r} \) variables are linearized, we enforce the feasibility condition via constraints~\ref{eq:main1}--\ref{eq:main4}, which correspond respectively to the \( (1,1) \) entry, the entries in the first row and column, the diagonal entries, and all remaining off-diagonal entries.

To ensure that the first column of \( \mathbf{P}' \) is fixed to \( \mathds{1} \), we impose constraint~\ref{eq:firstcol1}. To remove symmetry from the formulation, we apply lexicographic ordering to the rows of \( \mathbf{P}' \) using constraints~\ref{eq:rownumcalc}--\ref{eq:rownumlex}. This eliminates redundancy due to row permutations in \( \mathbf{P} \), which would otherwise yield equivalent solutions when paired with a corresponding reordering of the diagonal of \( \mathbf{W} \).

To express the objective \( \|\mathbf{W}\|_0 \), we introduce binary variables \( b_r \in \{0,1\} \), where \( \mathbf{W}_{r,r} \neq 0 \) implies \( b_r = 1 \). Constraint~\ref{eq:bw} enforces this encoding. The overall objective is then to minimize \( \sum_{r=1}^{k} b_r \), consistent with the formulation of Rajakumar et al.~\cite{rajakumar2022generating}.

To improve computational performance, we incorporate several valid inequalities and cuts:
\begin{itemize}
    \item \textbf{Row sum cut}: Because the first column of \( \mathbf{P}' \) is \( \mathds{1} \), only the first row may contain exactly one entry equal to 1. All other rows must contain at least two ones. This is enforced by constraint~\ref{const:more2}.
    
    \item \textbf{Padberg inequalities}~\cite{padberg1989boolean}: For any three binary variables \( a, b, c \), the inequality \( a + b + c - ab - ac - bc \leq 1 \) is valid. Replacing \( a, b, c \) with appropriate \( \mathbf{P}'_{r,i}, \mathbf{P}'_{r,j}, \mathbf{P}'_{r,k} \) terms yields the Padberg cut in constraint~\ref{eq:padberg1}. Although higher-order versions of these inequalities exist, they typically yield diminishing returns relative to their computational cost.
    
    \item \textbf{Spectral lower bound}: Corollary~\ref{cor:geommult} provides a spectral lower bound on \( \|\mathbf{W}\|_0 \). This is enforced via constraint~\ref{eq:lbcut}.
\end{itemize}

We also warm-start the solver using the Union of Stars construction, in line with best practices outlined in~\cite{klotz2013practical}. This significantly improves the performance of Gurobi’s built-in heuristics. In our formulation, the value of \( R \) corresponds to the number of rows in the warm-start solution. Although the Union of Double-Stars construction yields tighter upper bounds in theory, it performs worse in practice with Gurobi.

Lastly, we discuss the choice of the big-\( M \) constant. Rajakumar et al.~\cite[Theorem~8]{rajakumar2022generating} prove that \( |\mathbf{W}_{r,r}| \leq (3n - 2)^{(3n - 1)/2} \) suffices. However, in all our experiments, setting \( M = 10 \) produced good results with no indication of excluding feasible solutions.

As a post-processing step, we observed that constraint~\ref{eq:main1} is tautologically satisfied and may be removed. The same holds for the first row and column constraints~\ref{eq:main2}.

\newpage 
\begin{align}
    \min\quad & \sum_{r = 1}^{R} b_r \tag{1} \\
    \text{s.t.}\quad
    & -M b_r \le \mathbf{W}_{r,r} \le M b_r && \forall r \in [R] \tag{2} \label{eq:bw} \\
    & z_{i,j,r} \le \mathbf{P}'_{r,i} && \forall i,j \in [n], i \le j,\ r \in [R] \tag{3} \label{eq:and1} \\
    & z_{i,j,r} \le \mathbf{P}'_{r,j} && \forall i,j \in [n], i \le j,\ r \in [R] \tag{4} \label{eq:and2} \\
    & z_{i,j,r} \ge \mathbf{P}'_{r,i} + \mathbf{P}'_{r,j} - 1 && \forall i,j \in [n], i \le j,\ r \in [R] \tag{5} \label{eq:and3} \\
    & -M z_{i,j,r} \le t_{i,j,r} \le M z_{i,j,r} && \forall i,j \in [n], i \le j,\ r \in [R] \tag{6} \label{eq:binreal1} \\
    & \mathbf{W}_{r,r} - M(1 - z_{i,j,r}) \le t_{i,j,r} \le \mathbf{W}_{r,r} + M(1 - z_{i,j,r}) && \forall i,j \in [n], i \le j,\ r \in [R] \tag{7} \label{eq:binreal2} \\
    & \mathsf{tr} = \sum_{r = 1}^{R} \mathbf{W}_{r,r} \tag{8} \\
    & 4 \sum_{r = 1}^{R} t_{1,1,r} = 4\mathsf{tr} \tag{9} \label{eq:main1} \\
    & 4 \sum_{r = 1}^{R} t_{1,j,r} = 2\mathbf{A}_{1,j} + 2\mathsf{tr} && \forall j \in \{2, \dots, n\} \tag{10} \label{eq:main2} \\
    & 4 \sum_{r = 1}^{R} t_{i,i,r} = \mathbf{A}_{i,i} + 2\mathsf{tr} + 2\mathbf{A}_{1,i} && \forall i \in \{2, \dots, n\} \tag{11} \label{eq:main3} \\
    & 4 \sum_{r = 1}^{R} t_{i,j,r} = \mathbf{A}_{i,j} + \mathsf{tr} + \mathbf{A}_{1,i} + \mathbf{A}_{1,j} && \forall i,j \in \{2, \dots, n\},\ i < j \tag{12} \label{eq:main4} \\
    & \mathbf{P}'_{r,1} = 1 && \forall r \in [R] \tag{13} \label{eq:firstcol1} \\
    & \mathsf{rn}_r = \sum_{i = 1}^{n} 2^{i-1} \cdot \mathbf{P}'_{r,i} && \forall r \in [R] \tag{14} \label{eq:rownumcalc} \\
    & \mathsf{rn}_r \le \mathsf{rn}_{r+1} - 2 && \forall r \in [R-1] \tag{15} \label{eq:rownumlex} \\
    & \sum_{i = 1}^{n} \mathbf{P}'_{r,i} \ge 2 && \forall r \in \{2, \dots, R\} \tag{16} \label{const:more2} \\
    & \mathbf{P}'_{r,i} + \mathbf{P}'_{r,j} + \mathbf{P}'_{r,k} - z_{i,j,r} - z_{i,k,r} - z_{j,k,r} \le 1 && \forall i,j,k \in [n], i < j < k,\ r \in [R] \tag{17} \label{eq:padberg1} \\
    & \sum_{r = 1}^{R} b_r \ge n - \max_{\lambda \in \text{spec}(\mathbf{A})} \operatorname{mult}(\lambda) \tag{18} \label{eq:lbcut} \\
    & \mathbf{P}'_{r,i},\ z_{i,j,r},\ b_r \in \{0, 1\} && \forall r \in [R],\ i, j \in [n],\ i \le j \tag{19}
\end{align}

\section{Results \label{sec:exp}}

All experiments were conducted on a Dell Precision 3660 Tower equipped with an Intel Core i7-12700 processor (12 cores) and 32~GB of RAM. We used Gurobi 11.0.3 as the MIP solver. Both the proposed CMIPGC model and the baseline formulation of Rajakumar et al.~\cite{rajakumar2022generating} were run with a time limit of 3600 seconds. For CMIPGC, we disabled flow-cover cuts and MIR cuts, and set the heuristic intensity parameter to 0.005 to prioritize model-driven exploration.

We evaluated the models on 34 Erdős–Rényi graphs with vertex counts ranging from 4 to 20 (two graphs per vertex count). These instances are consistent with those used in Figure~\ref{fig:union-stars-double}, where we compare Union of Stars and Double-Star constructions. Graphs 14-1 and 18-2 were excluded from the results due to being empty.

The comparison of CMIPGC and the baseline model is shown in Table~\ref{tab:MIP-comp}. The baseline formulation has a size of $\mathcal{O}(2^n)$ in both variables and constraints, while CMIPGC grows polynomially with $\mathcal{O}(n^3)$. As expected, the baseline model performs well on small instances, but becomes intractable as $n$ increases. In contrast, CMIPGC scales more gracefully and consistently outperforms the baseline on larger graphs with more vertices and edges.

The full implementation and dataset used in our experiments are publicly available at our \href{https://github.com/SaberDinp/GraphCoupling.git}{GitHub repository}.

\begin{table}
    \centering
    \scriptsize
    \begin{tabular}{lcc|cccc|cccc}
        \hline
        Graph & \#nodes & \#edges & \multicolumn{4}{c}{Baseline} & \multicolumn{4}{c}{CMIPGC} \\
        \cline{4-7} \cline{8-11}
        & & & Primal bound & Dual bound & Gap(\%) & Time(s) & Primal bound & Dual bound & Gap(\%) & Time(s) \\
        \hline
4-1 & 4 & 3 & 5.0 & 5.0 & 0.0 & \textbf{0.0} & 5.0 & 5.0 & 0.0 & 0.05 \\
4-2 & 4 & 2 & 2.0 & 2.0 & 0.0 & \textbf{0.0} & 2.0 & 2.0 & 0.0 & 0.01 \\
5-1 & 5 & 6 & 4.0 & 4.0 & 0.0 & \textbf{0.01} & 4.0 & 4.0 & 0.0 & 0.45 \\
5-2 & 5 & 10 & 1.0 & 1.0 & 0.0 & \textbf{0.0} & 1.0 & 1.0 & 0.0 & 0.01 \\
6-1 & 6 & 8 & 7.0 & 7.0 & 0.0 & \textbf{0.77} & 7.0 & 7.0 & 0.0 & 1516.89 \\
6-2 & 6 & 6 & 6.0 & 6.0 & 0.0 & \textbf{0.17} & 6.0 & 6.0 & 0.0 & 107.2 \\
7-1 & 7 & 2 & 6.0 & 6.0 & 0.0 & \textbf{0.52} & 6.0 & 6.0 & 0.0 & 94.78 \\
7-2 & 7 & 17 & 6.0 & 6.0 & 0.0 & \textbf{0.95} & 6.0 & 6.0 & 0.0 & 236.93 \\
8-1 & 8 & 19 & 9.0 & 7.0 & 22.2 & 3600 & 9.0 & 7.0 & 22.2 & 3600 \\
8-2 & 8 & 8 & 7.0 & \textbf{7.0} & \textbf{0.0} & \textbf{92.04} & 7.0 & 6.0 & 14.28 & 3600 \\
9-1 & 9 & 21 & \textbf{10.0} & 3.0 & 70 & 3600 & 12.0 & \textbf{8.0} & \textbf{33.33} & 3600 \\
9-2 & 9 & 21 & \textbf{10.0} & 3.0 & 70 & 3600 & 11.0 & \textbf{8.0} & \textbf{27.27} & 3600 \\
10-1 & 10 & 22 & \textbf{11.0} & 2.0 & 81.81 & 3600 & 13.0 & \textbf{9.0} & \textbf{30.76} & 3600 \\
10-2 & 10 & 11 & \textbf{10.0} & 3.0 & 70 & 3600 & 12.0 & \textbf{8.0} & \textbf{33.33} & 3600 \\
11-1 & 11 & 51 & \textbf{8.0} & 3.0 & 62.5 & 3600 & 9.0 & \textbf{7.0} & \textbf{22.22} & 3600 \\
11-2 & 11 & 11 & \textbf{10.0} & 3.0 & 70 & 3600 & 11.0 & \textbf{8.0} & \textbf{27.27} & 3600 \\
12-1 & 12 & 16 & 28.0 & 2.0 & 92.85 & 3600 & \textbf{19.0} & \textbf{11.0} & \textbf{42.10} & 3600 \\
12-2 & 12 & 31 & 50.0 & 2.0 & 96 & 3600 & \textbf{17.0} & \textbf{11.0} & \textbf{35.29} & 3600 \\
13-1 & 13 & 56 & 49.0 & 1.0 & 97.95 & 3600 & \textbf{20.0} & \textbf{12.0} & \textbf{40} & 3600 \\
13-2 & 13 & 73 & \textbf{10.0} & 2.0 & 80 & 3600 & 11.0 & \textbf{7.0} & \textbf{36.36} & 3600 \\
14-2 & 14 & 6 & 15.0 & 2.0 & 86.66 & 3600 & \textbf{14.0} & \textbf{8.0} & \textbf{42.85} & 3600 \\
15-1 & 15 & 82 & 92.0 & 1.0 & 98.91 & 3600 & \textbf{37.0} & \textbf{14.0} & \textbf{62.16} & 3600 \\
15-2 & 15 & 31 & 89.0 & 1.0 & 98.87 & 3600 & \textbf{37.0} & \textbf{14.0} & \textbf{62.16} & 3600 \\
16-1 & 16 & 49 & 111.0 & 1.0 & 99.09 & 3600 & \textbf{27.0} & \textbf{15.0} & \textbf{44.44} & 3600 \\
16-2 & 16 & 118 & \textbf{4.0} & 1.0 & 75 & 3600 & 5.0 & \textbf{3.0} & \textbf{40} & 3600 \\
17-1 & 17 & 94 & 121.0 & 1.0 & 99.17 & 3600 & \textbf{41.0} & \textbf{16.0} & \textbf{60.97} & 3600 \\
17-2 & 17 & 51 & 128.0 & 1.0 & 99.21 & 3600 & \textbf{41.0} & \textbf{16.0} & \textbf{60.97} & 3600 \\
18-1 & 18 & 132 & 108.0 & 1.0 & 99.07 & 3600 & \textbf{42.0} & \textbf{15.0} & \textbf{64.28} & 3600 \\
19-1 & 19 & 150 & 143.0 & 1.0 & 99.30 & 3600 & \textbf{21.0} & \textbf{15.0} & \textbf{28.57} & 3600 \\
19-2 & 19 & 151 & 108.0 & 1.0 & 99.07 & 3600 & \textbf{32.0} & \textbf{13.0} & \textbf{59.37} & 3600 \\
20-1 & 20 & 57 & inf & 0.0 & inf & 3600 & \textbf{48.0} & \textbf{19.0} & \textbf{60.41} & 3600 \\
20-2 & 20 & 7 & inf & 1.0 & inf & 3600 & \textbf{11.0} & \textbf{8.0} & \textbf{27.27} & 3600 \\
        \hline
    \end{tabular}
    \caption{Comparison of CMIPGC with baseline MIP for 32 Erdős–Rényi graphs.}
    \label{tab:MIP-comp}
\end{table}

\section{Conclusion \label{sec:conc}}

In this work, we focused on the Graph Coupling problem for unweighted graphs. We improved the combinatorial construction of Rajakumar et al.~\cite{rajakumar2022generating}, reducing the upper bound on the graph coupling number from \(3n - 2\) to \(2.5n + 2\) for any graph with \(n\) vertices. Furthermore, we established the order-optimality of both constructions by identifying a family of graphs—namely, path graphs—for which the graph coupling number grows linearly with the number of vertices.

Beyond these theoretical results, we derived tighter bounds for specific graph families, which may provide insight into future constructions applicable to general weighted graphs. We also introduced a new compact mixed integer programming formulation (CMIPGC) with polynomial size in the input graph. This formulation outperforms the exponential-size baseline MIP of Rajakumar et al.~\cite{rajakumar2022generating} on larger graphs. Our approach leverages a strong linear-algebraic lower bound as a cutting plane and uses the combinatorial construction as a warm start, leading to substantial gains in Gurobi’s heuristic performance.

This line of research opens several avenues for future exploration:
\begin{itemize}
    \item Empirically, the all-ones row appears in some optimal matrix \( \mathbf{P} \) for all tested unweighted graphs. Proving that such a row is always present would imply symmetry in the graph coupling number, i.e., \( gc(G) = gc(\overline{G}) \).
    \item The diagonal entries of optimal matrices \( \mathbf{W} \) consistently fall within a small set of rational values such as \( \pm 1 \), \( \pm \tfrac{3}{4} \), \( \pm \tfrac{1}{2} \), \( \pm \tfrac{1}{4} \), and 0, suggesting that a much smaller upper bound for the big-\(M\) constant could suffice in practice.

    \item Extending the graph coupling framework to weighted graphs remains a promising direction.
    \item Both MIP formulations exhibit slow convergence in terms of improving bounds, likely due to weak LP relaxations. Strengthening these relaxations could lead to significantly faster solution times.
    \item Finally, the computational structure of the problem—minimizing the \(\ell_0\) norm under linear constraints—resembles compressive sensing. Investigating the computational complexity of the Graph Coupling problem and determining whether it is NP-hard is an important open question.
\end{itemize}

\section{Acknowledgment}
This work is partially funded by Office of Naval Research grant ONR-N000142412648.
\appendix
\section{Proofs of Theorems \ref{thm:gc=2} and \ref{thm:gc!=3} \label{sec:appen-proofs}}

\begin{proof}[\textbf{Proof of Theorem \ref{thm:gc=2}}]
\noindent
\textbf{($\Longleftarrow$)}\quad
Suppose \( G = K_a \cup K_b \) for some \( (a, b) \neq (1, 1) \). Then we can construct \( G \) using two spin bicliques: \( SB_{a, b} \) with weight \( \frac{1}{2} \), and \( SB_{a+b, 0} \) with weight \( \frac{1}{2} \). Thus, \( gc(G) \le 2 \). Since at least one of \( a \) or \( b \) is greater than 1, the graph is neither complete nor empty. By Theorem~\ref{thm:gc=1}, it follows that \( gc(G) = 2 \).

Similarly, if \( G = K_{a, b} \) for some \( (a, b) \neq (1, 1) \), we can use \( SB_{a, b} \) with weight \( \frac{-1}{2} \) and \( SB_{a+b, 0} \) with weight \( \frac{1}{2} \) to construct \( G \). Again, since \( gc(G) > 1 \), we conclude \( gc(G) = 2 \).

\smallskip
\noindent
\textbf{($\Longrightarrow$)}\quad
Assume \( gc(G) = 2 \). Then \( G \) can be expressed as a weighted sum of two spin bicliques with weights \( w_1 \) and \( w_2 \). Each spin biclique partitions the vertex set into two subsets; hence, two such bicliques partition the vertex set into four distinct regions.

Let the first biclique with weight \( w_1 \) induce the partition \( \{ \text{reg}_1 \cup \text{reg}_3, \text{reg}_2 \cup \text{reg}_4 \} \), and the second with weight \( w_2 \) induce \( \{ \text{reg}_1 \cup \text{reg}_2, \text{reg}_3 \cup \text{reg}_4 \} \). The edge weights between regions are illustrated in Figure~\ref{fig:gc=2}, where the weight between vertices within the same region is \( w_1 + w_2 \), by Corollary~\ref{cor:colclass}.

\begin{figure}[H]
    \centering
    \includegraphics[width=0.5\linewidth]{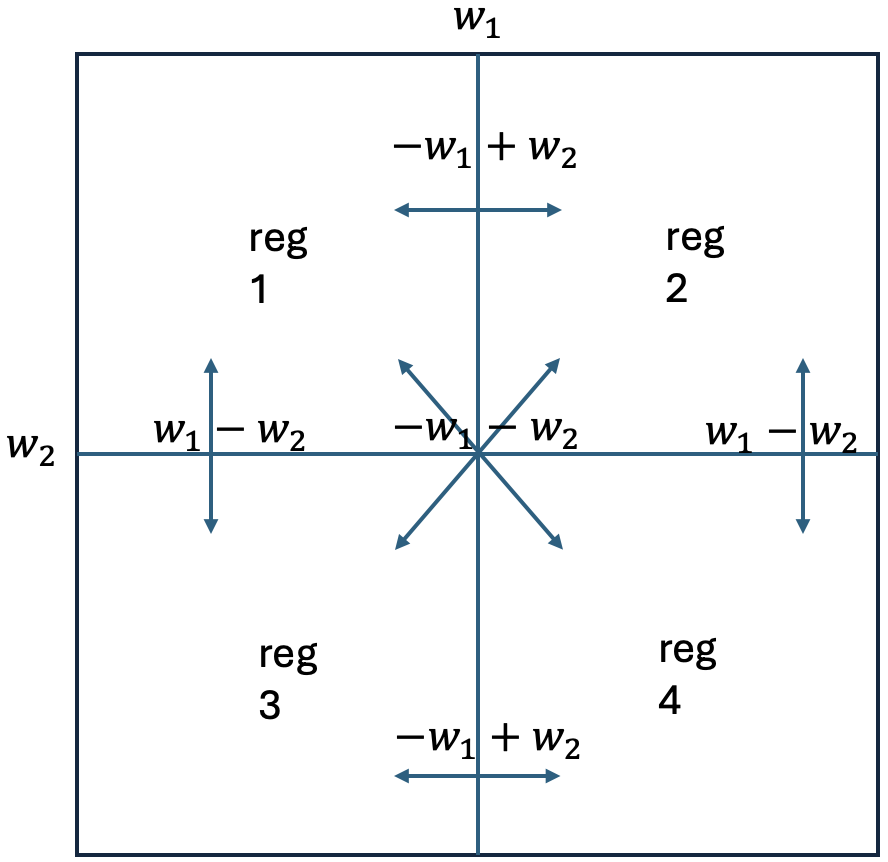}
    \caption{Each SB partitions the vertex set into two parts; thus two SBs induce four distinct regions. Edge weights between regions are shown in terms of \( w_1 \) and \( w_2 \). The weight of edges within a region is \( w_1 + w_2 \) as given in Corollary~\ref{cor:colclass}.}
    \label{fig:gc=2}
\end{figure}

We consider all possible non-empty region configurations:

\begin{itemize}
    \item \textbf{Single non-empty region:} If only one region is non-empty, the two SBs are identical and can be merged into one, contradicting \( gc(G) = 2 \).
    
    \item \textbf{Two non-empty regions:} If the two regions are opposite (e.g., reg\(_1\) and reg\(_4\)), the SBs are again identical. The only non-trivial case is when two adjacent regions are non-empty—without loss of generality, assume reg\(_1\) and reg\(_2\).
    \begin{itemize}
        \item If both contain only one vertex, the graph is either \( K_2 \) or \( \overline{K_2} \), both of which have \( gc = 1 \).
        \item If at least one region has more than one vertex, possible edge weights are \( w_1 + w_2 \in \{0, 1\} \) and \( -w_1 + w_2 \in \{0, 1\} \), resulting in graphs \( K_n \), \( \overline{K_n} \), \( K_{a,b} \), or \( K_a \cup K_b \). Since the first two have \( gc = 1 \), the valid cases are bicliques or unions of cliques.
    \end{itemize}

    \item \textbf{Three non-empty regions:} Without loss of generality, assume reg\(_1\), reg\(_2\), and reg\(_3\) are non-empty. This implies \( w_1 - w_2 = w_2 - w_1 = 0 \), i.e., \( w_1 = w_2 \).
    \begin{itemize}
        \item If each region contains a single vertex, we get graphs like \( K_2 \cup K_1 \) or \( \overline{K_3} \); the latter has \( gc = 1 \).
        \item If one region has more than one vertex, then \( w_1 = w_2 \) and \( w_1 + w_2 = 0 \), implying \( w_1 = w_2 = 0 \), and the graph is empty—contradicting \( gc(G) = 2 \).
    \end{itemize}

    \item \textbf{All four regions are non-empty:}
    \begin{itemize}
        \item If each region contains one vertex, then \( w_1 = w_2 \) and \( -w_1 - w_2 \in \{0, 1\} \). If \( -w_1 - w_2 = 0 \), then the graph is empty (contradiction). If \( -w_1 - w_2 = 1 \), then \( w_1 = w_2 = -\frac{1}{2} \), resulting in \( K_2 \cup K_2 \).
        \item If at least one region has more than one vertex, then the edge weight conditions imply \( w_1 = w_2 = 0 \), which again corresponds to an empty graph—contradicting \( gc(G) = 2 \).
    \end{itemize}
\end{itemize}

In all valid cases, \( G \) must be either a biclique \( K_{a,b} \) or a disjoint union \( K_a \cup K_b \), with \( a, b \ge 1 \) and \( (a, b) \neq (1, 1) \), completing the proof.
\end{proof}

\begin{proof}[\textbf{Proof of Theorem \ref{thm:gc!=3}}]

Assume for the sake of contradiction that \( gc(G) = 3 \). Then the three optimal spin bicliques induce a 3D hypercube structure, partitioning the vertex set into 8 distinct regions corresponding to the corners of a cube. We analyze all possible cases where between 1 and 8 corners are occupied and demonstrate that each case leads to a contradiction.

\begin{itemize}
    \item \textbf{1 or 2 corners occupied:} In this case, at least one SB is redundant, since two corners lie on a plane of the cube, and two of the SBs must be identical. Thus, the graph can be represented using only two SBs, contradicting \( gc(G) = 3 \).

    \item \textbf{3 corners occupied:} The only graphs with fewer than four vertices are \( K_1, K_2, \overline{K_2}, K_3, \overline{K_3}, K_1 \cup K_2 \), and \( K_{1,2} \), all of which have \( gc \leq 2 \). Thus, any graph requiring three spin bicliques must have at least four vertices. By the pigeonhole principle, if three corners of the cube are occupied by more than three vertices in total, then at least one corner must contain more than one vertex. Up to rotation, there are exactly three distinct configurations of three occupied corners:

    \begin{itemize}
        \item \textbf{Configuration (1, 1, 2)} (Fig.~\ref{fig:3v1-3d}): Corresponds to the linear system~\eqref{eq:system1}.
        \begin{figure}[H]
            \centering
            \includegraphics[width=0.2\textwidth]{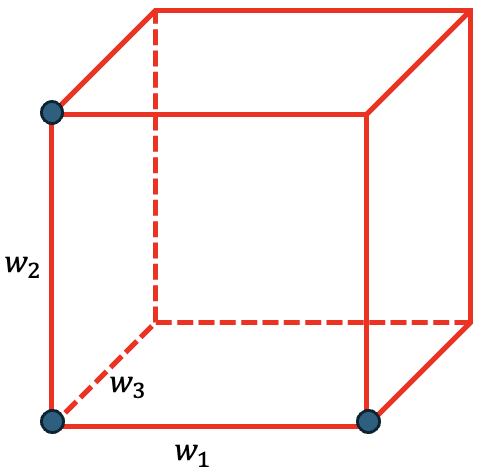}
            \caption{Three occupied corners with Manhattan distances 1, 1, 2.}
            \label{fig:3v1-3d}
        \end{figure}

        \begin{equation}
            \begin{bmatrix}
                1 & 1 & 1 \\
                -1 & 1 & 1 \\
                1 & -1 & 1 \\
                -1 & -1 & 1 \\
            \end{bmatrix}
            \begin{bmatrix}
                w_1 \\
                w_2 \\
                w_3 \\
            \end{bmatrix}
            =
            \begin{bmatrix}
                0 \,|\, 1 \\
                0 \,|\, 1 \\
                0 \,|\, 1 \\
                0 \,|\, 1 \\
            \end{bmatrix}
            \label{eq:system1}
        \end{equation}

        The coefficient matrix is full column rank, and brute force over all 16 RHS combinations shows that no solution \( \mathbf{w} \) exists with all entries nonzero. Thus, this configuration is forbidden.

        \item \textbf{Configuration (1, 2, 3)} (Fig.~\ref{fig:3v2-3d}): Corresponds to system~\eqref{eq:system2}.
        \begin{figure}[H]
            \centering
            \includegraphics[width=0.2\textwidth]{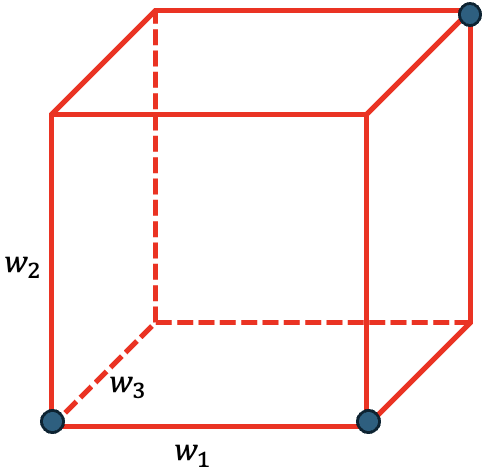}
            \caption{Three occupied corners with Manhattan distances 1, 2, 3.}
            \label{fig:3v2-3d}
        \end{figure}

        \begin{equation}
            \begin{bmatrix}
                1 & 1 & 1 \\
                -1 & 1 & 1 \\
                1 & -1 & -1 \\
                -1 & -1 & -1 \\
            \end{bmatrix}
            \begin{bmatrix}
                w_1 \\
                w_2 \\
                w_3 \\
            \end{bmatrix}
            =
            \begin{bmatrix}
                0 \,|\, 1 \\
                0 \,|\, 1 \\
                0 \,|\, 1 \\
                0 \,|\, 1 \\
            \end{bmatrix}
            \label{eq:system2}
        \end{equation}

        Brute force verifies that no such \( \mathbf{w} \) with all nonzero entries satisfies any RHS other than the trivial all-zeros (which corresponds to the empty graph with \( gc = 1 \)). Thus, this configuration is also forbidden.

        \item \textbf{Configuration (2, 2, 2)} (Fig.~\ref{fig:3v3-3d}): Corresponds to system~\eqref{eq:system3}.
        \begin{figure}[H]
            \centering
            \includegraphics[width=0.2\textwidth]{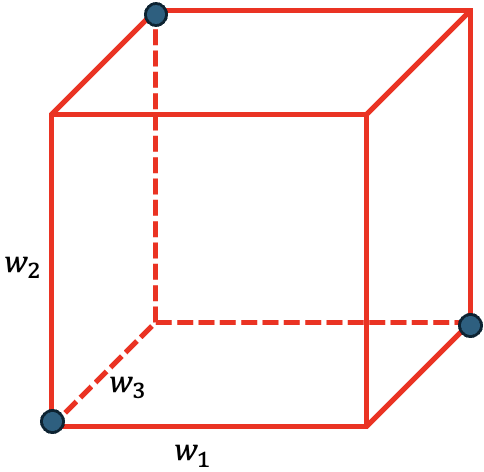}
            \caption{Three occupied corners with pairwise Manhattan distances 2.}
            \label{fig:3v3-3d}
        \end{figure}

        \begin{equation}
            \begin{bmatrix}
                1 & 1 & 1 \\
                -1 & -1 & 1 \\
                -1 & 1 & -1 \\
                1 & -1 & -1 \\
            \end{bmatrix}
            \begin{bmatrix}
                w_1 \\
                w_2 \\
                w_3 \\
            \end{bmatrix}
            =
            \begin{bmatrix}
                0 \,|\, 1 \\
                0 \,|\, 1 \\
                0 \,|\, 1 \\
                0 \,|\, 1 \\
            \end{bmatrix}
            \label{eq:system3}
        \end{equation}

        Again, no feasible solution exists with all \( w_i \neq 0 \) for any nontrivial RHS, so this configuration is forbidden as well.
    \end{itemize}

    \item \textbf{4 corners occupied:} 
    \begin{itemize}
        \item Suppose at least one corner contains more than one vertex. If any dimension exhibits a 3-to-1 or 4-to-0 split among the occupied corners, then the configuration must include one of the previously excluded cases—either (1,1,2) or (2,2,2)—which are known to be forbidden. Therefore, the only admissible configuration is a (2,2) split along each dimension, corresponding to the layouts shown in Fig.~\ref{fig:4v-1} and Fig.~\ref{fig:4v-2}. However, Fig.~\ref{fig:4v-1} leads to a redundant spin biclique, while Fig.~\ref{fig:4v-2} includes the forbidden (2,2,2) structure. Hence, no valid configuration remains in this case.

        \item If each of the four corners contains exactly one vertex, then among the 11 possible graphs on 4 vertices, none has \( gc = 3 \), as verified by the brute-force Algorithm~\ref{alg:gcbruteforce}.
        \begin{figure}[H]
            \centering
            \includegraphics[width=0.2\linewidth]{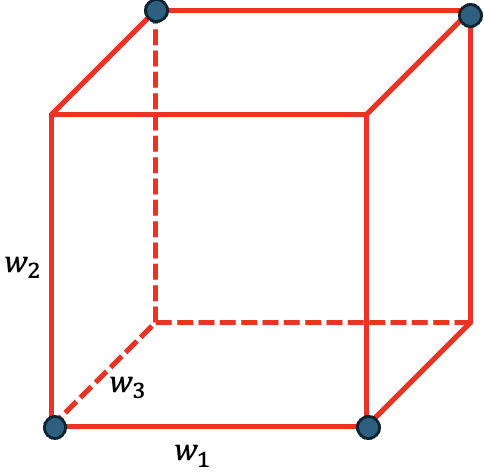}
            \caption{First (2,2)-split configuration in all dimensions. One SB is redundant.}
            \label{fig:4v-1}
        \end{figure}
        \begin{figure}[H]
            \centering
            \includegraphics[width=0.2\linewidth]{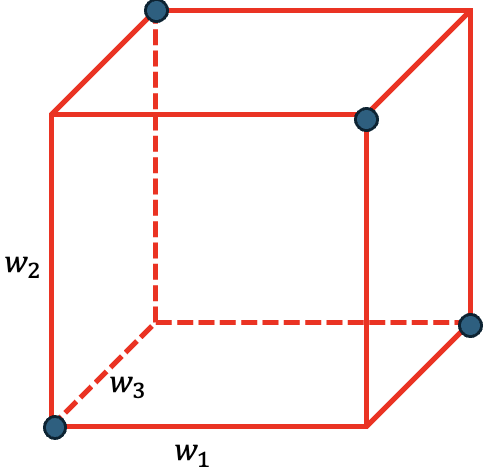}
            \caption{Second (2,2)-split configuration contains forbidden (2,2,2).}
            \label{fig:4v-2}
        \end{figure}
    \end{itemize}

    \item \textbf{5--8 corners occupied, with at least one corner containing multiple vertices:} Any such configuration necessarily contains one of the forbidden structures (1,1,2) or (2,2,2), hence is invalid.

    \item \textbf{5 corners occupied, each with one vertex:} We identify two new forbidden 4-corner configurations—Figs.~\ref{fig:4forbid1} and~\ref{fig:4forbid2}—which admit no solution to the corresponding linear systems for any of the 64 RHSs. Any configuration with 5 singleton-occupied corners is the complement of one of the 3-corner configurations previously discussed, and thus must contain one of these forbidden 4-corner structures.

    \begin{figure}[H]
        \centering
        \includegraphics[width=0.2\linewidth]{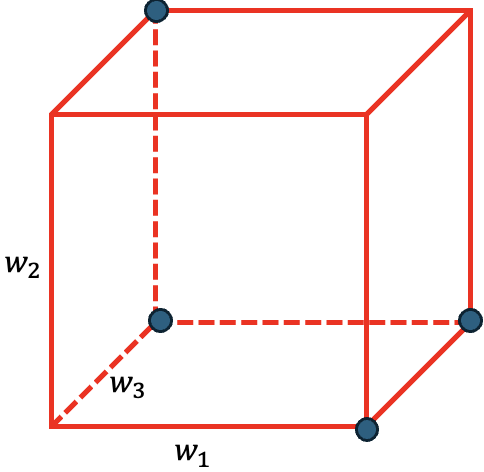}
        \caption{Forbidden structure 1: Four singleton corners with infeasible SB configuration.}
        \label{fig:4forbid1}
    \end{figure}
    \begin{figure}[H]
        \centering
        \includegraphics[width=0.2\linewidth]{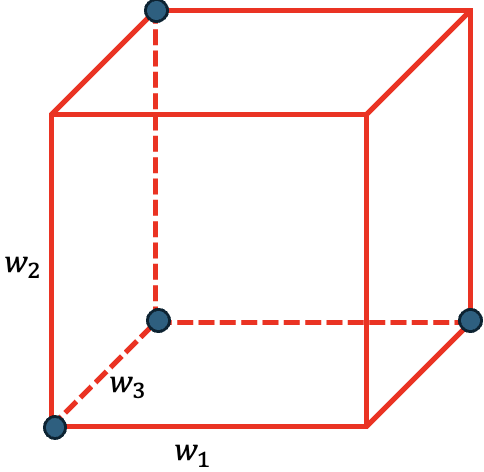}
        \caption{Forbidden structure 2: Another 4-corner configuration that cannot be covered by 3 SBs.}
        \label{fig:4forbid2}
    \end{figure}

    \item \textbf{6–8 singleton-occupied corners:} All such cases necessarily contain forbidden structure~\ref{fig:4forbid1} and are thus invalid.
\end{itemize}

In all possible configurations, we reach a contradiction to the assumption that \( gc(G) = 3 \). Therefore, there is no graph with graph coupling number equal to 3.
\end{proof}

\section{Brute Force Algorithm for Finding \texorpdfstring{$gc$}{gc} \label{sec:appen-bruteforce}}

The following is a straightforward brute-force algorithm for computing the graph coupling number \( gc(G) \). It is practical for small graphs, where full enumeration of possible \( \mathbf{P} \) matrices is computationally feasible. The algorithm incrementally tests increasing values of \( k \), generating all possible \( \{-1, 1\}^{k \times n} \) matrices with the first column fixed as \( \mathds{1} \), and checks whether a feasible weight vector \( \mathbf{w} \) exists such that \( \mathbf{P}^\top \mathbf{W} \mathbf{P} \odot \mathbf{J} = \mathbf{A} \).

\begin{algorithm}[H]
    \caption{Brute-Force Algorithm for Graph Coupling}
    \begin{algorithmic}[1]
        \State \textbf{Input:} Graph \( G = (V, E) \) with adjacency matrix \( \mathbf{A} \in \{0, 1\}^{n \times n} \)
        \State \( k \gets 1 \)
        \While{\( k \le 2.5n + 2 \)} 
            \ForAll{matrices \( \mathbf{P} \in \{-1, 1\}^{k \times n} \) with \( \mathbf{P}_1 = \mathds{1} \), enumerated in lexicographic row order}
                \State Initialize \( \mathbf{B} \in \mathbb{R}^{\binom{n}{2} \times k} \), \( \mathbf{c} \in \mathbb{R}^{\binom{n}{2}} \)
                \State \( r \gets 1 \)
                \For{\( i = 1 \) to \( n \)}
                    \For{\( j = i+1 \) to \( n \)}
                        \State \( \mathbf{B}_r \gets (\mathbf{P}_{i} \odot \mathbf{P}_{j})^\top \)
                        \State \( \mathbf{c}_r \gets \mathbf{A}_{i,j} \)
                        \State \( r \gets r + 1 \)
                    \EndFor
                \EndFor
                \If{the linear system \( \mathbf{B} \mathbf{w} = \mathbf{c} \) has a solution}
                    \State \Return \( (\mathbf{P}, \mathbf{w}) \) as an optimal solution
                \EndIf
            \EndFor
            \State \( k \gets k + 1 \)
        \EndWhile
    \end{algorithmic}
    \label{alg:gcbruteforce}
\end{algorithm}

\bibliographystyle{plain}
\bibliography{apstemplate}

\end{document}